\documentclass[journal]{IEEEtran}
\usepackage{amsmath,amssymb,amsfonts, mathrsfs}
\usepackage{amsthm, mathtools, bm, bbm, dsfont}
\usepackage[linesnumbered,ruled]{algorithm2e}
\usepackage{algorithmic,float}
\usepackage{array}
\usepackage{setspace}
\usepackage{hyperref}
\usepackage{graphicx}
\makeatletter
    \let\MYcaption\@makecaption
\makeatother
    \usepackage[font=footnotesize]{subcaption}
\makeatletter
    \let\@makecaption\MYcaption
\makeatother

\usepackage{tikz}
\usetikzlibrary{arrows.meta, positioning}

\usepackage{tabularx,booktabs}
\usepackage{color}
\usepackage[nocompress]{cite}
\usepackage{multirow}
\graphicspath{{images/}}
\usepackage{subfiles}

\allowdisplaybreaks

\DeclareMathOperator*{\argmin}{arg\,min}
\DeclareMathOperator{\rvi}{RVI}

\newtheorem{theorem}{Theorem}
\newtheorem{lemma}{Lemma}
\newtheorem{remark}{Remark}
\newtheorem{definition}{Definition}
\newtheorem{proposition}{Proposition}

\newtheorem{corollary}{Corollary}

\begin{document}

\title{Semantic-Aware Remote Estimation of Multiple Markov Sources Under Constraints}
\author{Jiping~Luo and Nikolaos~Pappas,~\IEEEmembership{Senior~Member,~IEEE}
\thanks{This work has been supported in part by the Swedish Research Council (VR), ELLIIT, the Graduate School in Computer Science (CUGS), the European Union (ETHER, 101096526 and 6G-LEADER, 101192080), and the European Union's Horizon Europe research and innovation programme under the Marie Skłodowska-Curie Grant Agreement No 101131481 (SOVEREIGN). This paper was presented in part at the IEEE PIMRC 2024 \cite{jipingPIMRC2024} and WiOpt 2024 \cite{jipingWiOpt2024}.}
\thanks{The authors are with the Department of Computer and Information Science, Link\"oping University, Link\"oping 58183, Sweden (e-mail: jiping.luo@liu.se; nikolaos.pappas@liu.se).}
}

\maketitle
\begin{abstract}
    This paper studies the remote estimation of multiple Markov sources over a lossy and rate-constrained channel. Unlike most existing studies that treat all source states equally, we exploit the \emph{semantics of information} and consider that the remote actuator has different tolerances for the estimation errors. We aim to find an optimal scheduling policy that minimizes the long-term \textit{state-dependent} costs of estimation errors under a transmission frequency constraint. The optimal scheduling problem is formulated as a \emph{constrained Markov decision process} (CMDP). We show that the optimal Lagrangian cost follows a piece-wise linear and concave (PWLC) function, and the optimal policy is, at most, a randomized mixture of two simple deterministic policies. By exploiting the structural results, we develop a new \textit{intersection search} algorithm that finds the optimal policy using only a few iterations. We further propose a reinforcement learning (RL) algorithm to compute the optimal policy without knowing \textit{a priori} the channel and source statistics. To avoid the ``curse of dimensionality" in MDPs, we propose an online low-complexity \textit{drift-plus-penalty} (DPP) algorithm. Numerical results show that continuous transmission is inefficient, and remarkably, our semantic-aware policies can attain the optimum by strategically utilizing fewer transmissions by exploiting the timing of the important information.
\end{abstract}

\section{Introduction}\label{sec:introcution}
\IEEEPARstart{R}{emote} estimation of stochastic processes is a fundamental problem in networked control systems (NCSs)\cite{Aditya2017, NikolaosGoalOriented}. However, such applications confront critical scalability challenges due to the scarcity of communication resources and the immensity of sensory data\cite{TIISurvey}. Therefore, there is a need to reduce the amount of data communicated between sensors, plants, and actuators. This is often done through various metrics that capture the \textit{value} of communicating and determine whether a new measurement has a positive impact on the control system \cite{gielis2022critical, Marios}. Despite significant research efforts, most existing communication protocols are context- and content-agnostic, leading to inefficient communications and degraded performance\cite{Marios}.

The primary objective has been to minimize the estimation error (i.e., the \textit{distortion} between the source state at the sensor side and the reconstructed state at the receiver\cite{Karl2013RE, leong2020deep, Aditya2017, pezzutto2022transmission, Cocco2023}). It indicates that information is valuable when it is accurate. However, in practice, this is not always the case, as \textit{the actuator may have different tolerance for estimation errors of different states or in different situations}. Consider, for example, a remotely controlled vehicle sending status messages to a central node and following instructions to ensure safe and successful operations. While inaccurate estimates of normal status are somewhat acceptable, estimate errors of emergent messages (e.g., off-track or close to obstacles) can cause significant safety hazards and must be avoided. Therefore, it is essential to factor into the communication process the \textit{semantics} (i.e., state-dependent significance, context-aware requirements, and goal-oriented usefulness) of status updates and prioritize the information flow efficiently\cite{Marios, Petar}. 

Information \textit{freshness}, measured by the Age of Information (AoI), that is, the time elapsed since the newest received update at the destination\cite{pappas2023age}, has recently been employed in NCSs\cite{WiSwarm, kutsevol2023experimental, jayanth2023distortion, VoI, AoI_onur, AoI_estimation}. It implies that status updates are more valuable when they are fresh. However, AoI alone cannot fully characterize the estimation error or meet the requirements of NCSs\cite{VoI, AoI_onur}. For linear time-invariant (LTI) systems, the error covariance, also known as the Value of Information (VoI)\cite{VoI}, is a monotonically non-decreasing function of AoI\cite{Karl2013RE,leong2020deep}. For Wiener processes, the discrepancy between the source signal and the newest received sample follows a Gaussian distribution with AoI as its covariance\cite{AoI_estimation}. Although AoI is closely related to estimation error in the above systems, the age-optimal policy is not optimal in minimizing the estimation error\cite{VoI, AoI_estimation}. Furthermore, AoI ignores the source evolution and the context of the application. 

Several metrics have been introduced to address the shortcomings of AoI and distortion\cite{AoII_TWC, AoII_GC, niko2019statechange, jayanth2023distortion, UoI_estimation, AoA}. Age of Incorrect Information (AoII)\cite{AoII_TWC, AoII_GC}, defined as a composite of distortion and age penalties, captures the cost of not having a correct estimate for some time. State-aware AoI \cite{niko2019statechange, jayanth2023distortion} accounts for the state-dependent significance of stochastic processes. A weighted distortion metric, namely Urgency of Information (UoI)\cite{UoI_estimation}, incorporates context awareness through adaptive weights. Age of Actuation (AoA) \cite{AoA} is a more general metric than AoI and becomes relevant when the information is utilized to perform actions in a timely manner. However, these metrics treat different estimation errors equally.

A new semantics-aware metric, namely Cost of Actuation Error (CAE), was first introduced in \cite{NikolaosGoalOriented} to account for state-dependent significance. Specifically, the CAE prioritizes information flow according to the significance of states, or more precisely, the potential control risks or costs incurred without correct estimates of these states. The problem of remote estimation of a single discrete-time finite-state Markov source was further studied in \cite{Salimnejad2023TCOM, Salimnejad2023JCN, EmmanouilMinizationCAE, zakerisemantic}. The work in \cite{Salimnejad2023TCOM, Salimnejad2023JCN} analyzed a low-complexity randomized policy that takes probabilistic actions independent of the source evolution and the real-time estimation errors. The optimal policy in resource-constrained systems was presented in \cite{EmmanouilMinizationCAE}. These works consider a single-source scenario, assuming that the source or channel statistics are known \textit{a priori}. However, acquiring such information before deployment may be costly or even not feasible. Theoretical results regarding the existence and structure of the optimal policy have not been developed.

This paper considers a general case where an agent schedules the update time of multiple sources over a lossy and rate-constrained channel. The system consists of \textit{slowly} and \textit{rapidly} evolving sources. We aim to find an optimal schedule to minimize average CAE under a transmission frequency constraint. We identify this problem as a \textit{constrained Markov Decision Process} (CMDP), which is proven to be computationally prohibitive\cite{altman1999constrained}. Moreover, finding an optimal policy in unknown environments is significantly more complicated. 

The main contributions of this paper are as follows:
\begin{itemize}
    \item We formulate the semantic-aware multi-source scheduling problem as an average-cost CMDP. Then, we relax it as an unconstrained two-layer problem using the \textit{Lagrangian technique}. Our analysis reveals that the optimal Lagrangian cost is a piece-wise linear and concave (PWLC) function, and the optimal policy is, at most, a randomized mixture of two simple deterministic policies. By exploiting these findings, we develop a new policy search method called \textit{intersection search} that can find the optimal policy using only a few iterations. We further propose a \textit{Q-learning} algorithm that learns an optimal policy for unknown source and channel statistics.
    \item To address the ``curse of dimensionality'' of MDPs, we propose a low-complexity online policy based on the \textit{Lyapunov optimization theorem}. We demonstrate that this policy satisfies the transmission frequency constraint and show its near-optimality.
    \item We conduct extensive simulations to validate the performance of the proposed policies. We show that continuous transmission is inefficient and yields minimal (even no) performance gains for the communication goal; instead, transmitting less but more significant information can also lead to an optimal outcome.
\end{itemize}

The rest of this paper is organized as follows. Sections~\ref{sec:system_model} and~\ref{sec:problem_formulation} introduce the system model and formulate the optimal scheduling problem. Section~\ref{sec:main_results} presents the main structural results. In Section~\ref{sec:solutions}, we develop algorithms for both known and unknown environments. Numerical results are provided in Section~\ref{sec:results}, and we conclude our work in Section~\ref{sec:conclusion}.

\begin{figure}[htbp]
    \centering
    \includegraphics[width=\linewidth]{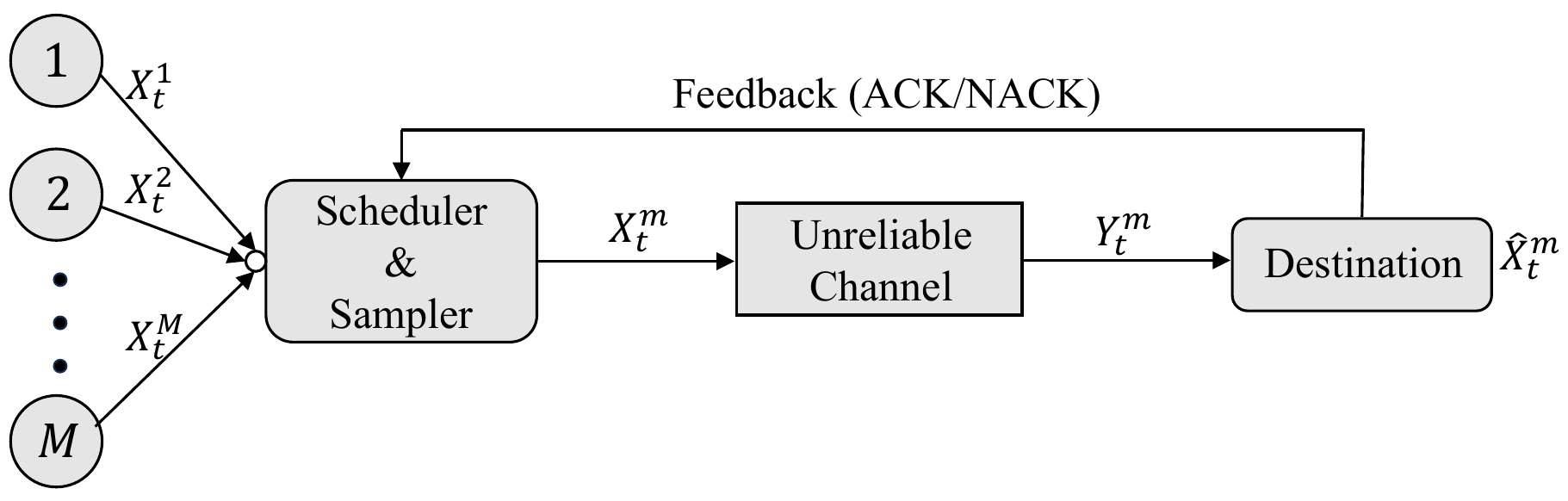}
    \vspace{-0.2in}
    \caption{Remote estimation of multiple Markov sources with feedback.}
    \label{fig:systemmodel}
\end{figure}

\section{System Model}\label{sec:system_model}
\subsection{System Description}
We consider a slotted-time communication system shown in Fig.~\ref{fig:systemmodel}, where the destination remotely tracks the states of multiple stochastic processes for achieving a specific control task. We illustrate each component as follows:
\subsubsection{Sources}
Denote $\mathcal{M} = \{1, 2, \ldots, M\}$ as the index set of the sources. Each source $m\in \mathcal{M}$ is modeled by a $N_m$-state homogeneous discrete-time Markov chain (DTMC)\footnote{DTMC is widely applied in safety-critical systems, such as autonomous driving\cite{althoff2009,ure2019} and cyber-physical security\cite{shi2017cyber}, for modeling operation modes, control risks, and working conditions. A prominent feature of these applications is the state-dependent and context-aware costs or risks.}, $\{X_t^m\}_{t\geq1}$, where $t = 1, 2, \ldots$ is the discrete time index. The slot length is defined as the interval between two successive state changes. The value of $X_t^m$ is chosen from a finite set $\mathcal{X}^m = \{1, 2, \ldots, N_m\}$. The state transition probabilities of each source $m$ are represented by a stochastic matrix $\mathbf{Q}^m$, where $\mathbf{Q}^m_{i,j} = \mathbb{P}(X_{t+1}^m = j|X_t^m = i)$ is the transition probability from state $i$ to state $j$ between two consecutive time slots. A DTMC is said to be \textit{slowly evolving} if it tends to stay in the same state across two consecutive time slots, i.e., $\mathbf{Q}$ is diagonally dominant. Otherwise, we call it \textit{rapidly evolving}. We consider both slowly and rapidly evolving sources. 

\subsubsection{Agent}
An agent observes the sources' realizations and decides at the beginning of each time slot which source to sample. After generating a status update of some source, the transmitter sends the packet to the destination over a wireless channel. We assume that only one packet can be sent at a time. Each packet contains the state of the source and a header with two fields: a timestamp and an index. The timestamp shows the generation time of the packet, and the index indicates the source from which the packet is sampled. Thus, the receiver knows exactly each packet's generation time and origin.

We denote the sampling decision at time $t$ by an integer variable $A_t\in \mathcal{A}= \{0, 1, \ldots, M\}$, where $A_t = 0$ means no source is selected and the transmitter remains salient, while $A_t = m, m\neq 0$ means taking a sample from source $m$ and transmitting it immediately. In practice, the transmitter can only transmit packets intermittently due to resource constraints and network contentions. So we impose a constraint on the transmission frequency that is no larger than $F_{\max}\in (0, 1]$. 

\subsubsection{Channel}
We consider an unreliable channel with \textit{i.i.d.} packet drops and (possibly) transmission delays. The channel state $H_t$ is $1$ if the packet is successfully decoded at the receiver and $0$ otherwise. We assume that the channel state information (CSI) is available only at the destination node. Therefore, the agent does not know whether a status update will reach the destination before a transmission attempt. We define the success transmission probability as $p_s = \mathbb{E}[H_t]$. In this paper, we consider two specific delay cases (see Fig.~\ref{fig:systemstate}):
\begin{itemize}
    \item \textit{Zero-delay}: The time duration of the sampling and transmission process is negligible compared to the slot length (i.e., the interval between two successive state changes). A sampled packet reaches its destination immediately after being generated.
  \item \textit{One-delay}: A transmission attempt occupies (nearly) one slot length. Thus, the receiver receives an update from the transmitter with a delay of one slot.
\end{itemize}
We denote zero-delay as $d = 0$ and one-delay as $d = 1$. While channel delays in practice are more complex and variable, these two cases provide a foundational understanding and allow for analytical tractability. Moreover, many practical systems operate under strict delay constraints, where zero- and one-delay approximations are reasonable. We refer readers to \cite{AoI_estimation} for an analysis of the impact of random delays on distortion and AoI.

\subsubsection{Receiver}
The receiver sends an acknowledgment (ACK) signal to the transmitter after each transmission to indicate the transmission success or failure. We assume that these feedback packets are delivered instantly and error-free. The receiver needs to reconstruct the states of the sources from under-sampled and (possibly) delayed measurements, and then forward them to the controller/actuator for control. We assume that the receiver does not know the sources' evolution patterns and updates its estimate using the latest received packet, i.e.,
\begin{align}
	\hat{X}_{t}^m = \begin{cases}
		X_{t-d}^m, &\text{if}~A_{t-d} = m, H_{t-d} = 1, d=0,1,\\
		\hat{X}_{t-1}^m, & \text{otherwise}.
	\end{cases} \label{eq:estimate}
\end{align}

\begin{figure}[t]
    \centering
    \scalebox{0.85}{\begin{tikzpicture}[>=stealth, scale=1.0]


\draw[->, thick] (0,0) -- (8,0) node[right] {};
\draw[->, thick] (0,-2) -- (8,-2) node[right] {};

\node[left] at (0,0) {Tx};
\node[left] at (0,-2) {Rx};

\draw[dashed] (.5,0.5) -- (.5,-2.5) node[below] {$t$};
\draw[dashed] (7,0.5) -- (7,-2.5) node[below] {$t+1$};

\draw (.8,0.) -- (.8,.2);
\node at (.8,0.4) {$X_t$};
\draw (1.2,0.) -- (1.2,.2);
\node at (1.2,0.4) {$A_t$};
\draw (1.8,0.) -- (1.8,.2);
\node at (1.8,0.4) {$\hat{X}_{t+1}$};

\draw[->] (1.2,0) -- (1.6,-2);
\node at (1.6,-2.3) {$\hat{X}_{t}$};
\draw[->, dashed] (1.6,-2) -- (1.8,0) node[pos=0.5, right] {\small ACK};

\node at (4.,-2.8) {(a)~zero-delay};


\draw[->, thick] (0,-4) -- (8,-4);
\draw[->, thick] (0,-6) -- (8,-6);

\node[left] at (0,-4) {Tx};
\node[left] at (0,-6) {Rx};

\draw[dashed] (.5,-3.5) -- (.5,-6.5) node[below] {$t$};
\draw[dashed] (7,-3.5) -- (7,-6.5) node[below] {$t+1$};

\draw (.8,-4.) -- (.8,-3.8);
\node at (.8,-3.6) {$X_t$};
\draw (1.2,-4.) -- (1.2,-3.8);
\node at (1.2,-3.6) {$A_t$};
\draw (6.8,-4.) -- (6.8,-3.8);
\node at (6.8,-3.6) {$\hat{X}_{t+1}$};

\draw[->] (1.2,-4) -- (6.6,-6);
\draw[->, dashed] (6.6,-6) -- (6.8,-4) node[pos=0.5, left] {\small ACK};

\node at (6.6,-6.3) {$\hat{X}_{t+1}$};

\node at (4.,-6.8) {(b)~one-delay};

\end{tikzpicture}}
    \vspace{-0.1in}
    \caption{An illustration of the zero- and one-delay cases.}
    \label{fig:systemstate}
\end{figure}

\subsubsection{Cost of Actuation Error}
Let $S_t$ denote the system state of the MDP. As shown in Fig.~\ref{fig:systemstate}, the agent observes the sources' states $X_t = (X^1_t, \ldots, X^M_t)$ at the beginning of time slot $t$ and then makes a decision $A_t$ based on the available information. In the zero-delay case, the agent has access to the previously reconstructed state $\hat{X}_{t-1}$, but does not know $\hat{X}_t$ at the time of decision-making. Therefore, the system state for the zero-delay case is defined as $S_t = (X_t, \hat{X}_{t-1})$. In contrast, in the one-delay case, the feedback signal was received at the end of slot $t-1$, and the receiver will update its estimate at the beginning of slot $t$ using either $X_{t-1}$ or $\hat{X}_{t-1}$. Therefore the agent knows $\hat{X}_t$ when making the decision $A_t$. So we define $S_t = (X_t, \hat{X}_t)$ for the one-delay case. We also define the state of subsystem $m$ at time $t$ as $S^m_t = (X^m_t, \hat{X}^m_{t-1})$ for zero-delay case and $S^m_t = (X^m_t, \hat{X}^m_{t})$ for one-delay case. 

The discrepancy between the source realizations and the reconstructed states can cause actuation errors. In practice, some source states are more important than others, and thus some actuation errors may have a larger impact than others. Let $\delta^m_{i,j} \triangleq \delta^m (X^m = i, \hat{X}^m = j)$ denote the potential control cost of the sub-system $m$ being at state $s^m = (i,j)$, where $\delta^m: \mathcal{X}^m\times\mathcal{X}^m\rightarrow\mathbb{R}_{\geq0}$ is a pre-determined non-commutative function\footnote{Function $\delta^m$ is non-commutative, i.e., for any $i, j \in \mathcal{X}^m$, $\delta^m_{i,j}$ is not necessarily equal to $\delta^m_{j,i}$. Moreover, we consider that $\delta^m_{i,j}$ needs not equal to $\delta^m_{i, k}$. This is because different actuation errors may have different repercussions for the system's performance.} capturing \textit{state-dependent significance} of estimation errors\cite{NikolaosGoalOriented}. Then, the \textit{cost of actuation error} (CAE) of taking an action in a certain system state is defined as the weighted sum of all subsystems' control costs
\begin{align}
	c(S_t, A_t, S_{t+1})
	=\sum\limits_{m \in \mathcal{M}} \omega_m  \delta^m(X_{t + d}^m, \hat{X}_{t + d}^m),\label{eq:CAE}
\end{align}
where $\omega_m$ represents the relative \textit{significance} of source $m$. 

\begin{remark}
    The cost function \eqref{eq:CAE} depends on the subsequent system state, so the agent can only evaluate the effectiveness of transmission attempts after receiving feedback packets. It implies that the agent should know a priori the expected CAE of taking an action in a certain state. Therefore, the agent must consider several factors in the decision-making process, such as the weight ($\omega_m$), the channel conditions ($d$ and $p_s$), the source evolution pattern ($\mathbf{Q}^m$), and the system resource constraint ($F_{\max}$) to make more informed decisions.
\end{remark}

\begin{remark}
    It is worth mentioning that, although the cost function in \eqref{eq:CAE} is decomposable, the results in this paper extend automatically to coupled collective costs. For example, in collaborative beamforming\cite{lipski2024age}, each transmitter compensates for its channel phase offsets using up-to-date channel knowledge such that the received signals from different transmitters add up coherently at the receiver to improve the overall beamforming gain. In this context, the discrepancy between the true channel phase $X(t)$ and the estimated state $\hat{X}(t)$ leads to incorrect phase adjustments and performance losses.
\end{remark}

\subsection{Performance Measures}
At each time $t$, the information available at the agent is $I_t = (S_{1:t}, A_{1:t-1})$. Let $\mathcal{I}_t$ and $\mathcal{S}$ denote the set of all information leading up to time $t$ and the system state space, respectively. A \textit{decision rule} is a function that maps the information set into the set of probability distributions on the action space, that is $\pi_t:\mathcal{I}_t\rightarrow \mathcal{P}(\mathcal{A})$. A decision rule is said to be \textit{Markovian} if it depends only on the current system state, that is $\pi_t:\mathcal{S}\rightarrow \mathcal{P}(\mathcal{A})$. Moreover, a decision rule is \textit{deterministic} if it chooses a unique action
in each state, i.e., $\pi_t:\mathcal{S}\rightarrow\mathcal{A}$. 

A \textit{policy} $\pi = \{\pi_t\}_{t=1}^\infty$ is a sequence of decision rules. We call a policy \textit{stationary} if the mapping $\pi_t$ does not depend on $t$. The taxonomy of decision rules leads to three classes of policies, ranging in generality from stationary Markovian deterministic (MD), stationary Markovian randomized (MR), to non-stationary and/or history-dependent. We denote these classes of policies as $\Pi^\textrm{MD}$, $\Pi^\textrm{MR}$, and $\Pi$, respectively. 

At each time $t$, the system incurs two costs: an actuation cost defined in \eqref{eq:CAE}, and a transmission cost given by
\begin{align}
	f(S_t, A_t) = \mathds{1}(A_t\neq 0).
\end{align}
We define the \textit{limsup\footnote{The \textit{lim average costs} need not exist for non-stationary policies\cite{puterman1994markov}.} average CAE} of a policy $\pi\in\Pi$ by
\begin{align}
	C_s(\pi) \triangleq \limsup_{T\rightarrow\infty}\frac{1}{T}\sum_{t=1}^{T}\mathbb{E}^\pi\biggl[c(S_t, A_t, S_{t+1})\bigg|S_1 = s\biggr],
\end{align}
where $\mathbb{E}^\pi$ represents the conditional expectation, given that policy $\pi$ is employed with initial state $s\in\mathcal{S}$. Similarly, the average transmission cost can be defined as
\begin{align}
	F_s(\pi) \triangleq \limsup_{T\rightarrow\infty}\frac{1}{T}\sum_{t=1}^{T}\mathbb{E}^\pi\biggl[f(S_t, A_t)\bigg|S_1 = s\biggr].
\end{align}

\subsection{System Dynamics}
Let $P:\mathcal{S}\times\mathcal{A} \rightarrow \mathcal{S}$ denote the system transition function, where $P(s^\prime|s,a) \coloneqq \mathbb{P}(S_{t+1} = s^\prime|S_t = s, A_t = a)$ is the probability that the system state becomes $s^\prime$ at time $t+1$ when the agent chooses action $a$ in state $s$ at time $t$. For a given sampling decision, each subsystem $m\in\mathcal{M}$ evolves independently. Thus, we have
\begin{align}
	P(s^\prime|s, a) = \prod_{m\in\mathcal{M}} P^m(s^\prime(m)|s(m), a),
\end{align}
where $P^m$ is the transition function of subsystem $m$, given by
\begin{align}
	&P^m(s^\prime(m)|s(m), a) \notag\\
	&=\begin{cases}
		\mathbf{Q}^m_{i,k}p_s, &\textrm{if}~a = m, s(m) =(i, j), s^\prime(m) = (k, i),\\ 
		\mathbf{Q}^m_{i,k}p_f, &\textrm{if}~a = m, s(m) =(i, j), s^\prime(m) = (k, j),\\
		\mathbf{Q}^m_{i,k}, &\textrm{if}~a = m, s(m) =(i, i), s^\prime(m) = (k, i),\\
		\mathbf{Q}^m_{i,k}, &\textrm{if}~a \neq m, s(m) = (i,j), s^\prime(m) = (k, j),\\
		\mathbf{Q}^m_{i,k}, &\textrm{if}~a \neq m, s(m) = (i,i), s^\prime(m) = (k, i),\\
		0,&\textrm{otherwise},
	\end{cases} \label{eq:Pm}
\end{align}
where $p_f = 1-p_s, i,j, k \in \mathcal{X}^m, i\neq j$. Note that Eq.~\eqref{eq:Pm} is common for both zero- and one-delay cases. 

\section{Problem Formulation and Analysis}\label{sec:problem_formulation}
\subsection{CMDP Formulation}
We aim to find a policy $\pi$ in space $\Pi$ that minimizes the average CAE while satisfying a transmission frequency constraint. We can formulate this as an optimization problem
\begin{align}
	\mathscr{P}_\textrm{CMDP}:\,\, \inf_{\pi\in\Pi} C_s(\pi), \,\,\,\textrm{s.t.}\,\, F_s(\pi) \leq F_{\max}.\label{problem:CMDP}
\end{align}
The above problem is essentially an \textit{average-cost constrained Markov Decision Process} (CMDP), which is, however, challenging to solve since it imposes a \textit{global} constraint that involves the entire decision-making process. Notably, this transmission frequency constraint can enhance system scalability. In large-scale systems, sources can be grouped into smaller clusters, each utilizing only a small portion of the spectrum resource. This approach reduces resource contention and allows the system to scale more effectively.

\begin{definition}
    (Globally feasible policy) A policy $\pi$ is called globally feasible, or simply, feasible, if it satisfies the transmission constraint for all initial states, i.e., $F_s(\pi) \leq F_{\max}$ for all initial state $s\in\mathcal{S}$. 
\end{definition}

\begin{definition} (Constrained optimal policy)
	    Any feasible policy that attains $C_s^* = \inf_{\pi\in\Pi} C_s(\pi)$ for each $s\in\mathcal{S}$ is termed a (constrained) optimal policy.
\end{definition}

\subsection{Average Costs of Stationary Policies}
There is no loss of optimality in restricting attention to stationary policies\cite[Proposition~6.2.3]{sennott1998stochastic}. Under a stationary Markovian policy $\pi$, the induced stochastic process $\{S_t\}_{t\geq1}$ is a time-homogeneous Markov chain with transition matrix $\mathbf{P}(\pi)$. The $(s, s^\prime)$-th entry of $\mathbf{P}(\pi)$ is given by
\begin{align}
	\mathbf{P}_{s,s^\prime}(\pi) = \begin{cases}
		P(s^\prime|s, \pi(s)), & \textrm{if}~\pi\in \Pi^\textrm{MD},\\
		\sum\limits_{a\in\mathcal{A}}P(s^\prime|s,a)\pi(a|s), & \textrm{if}~\pi\in \Pi^\textrm{MR}.
	\end{cases}
\end{align}
Since $\mathbf{P}(\pi)$ is stationary, the probability of the system transitioning to state $s^\prime$ at time $t$, given an initial state of $s$, can be represented as $\mathbf{P}^{t-1}_{s,s^\prime}(\pi)$.

To simplify notation, we denote the expected CAE of a policy $\pi\in\Pi^\textrm{MR}$ in state $s$ as
\begin{align}
	\bm{c}^\pi(s) = \begin{cases}
		\sum\limits_{s^\prime\in\mathcal{S}}c(s, \pi(s), s^\prime)P(s^\prime|s, \pi(s)), & \textrm{if}~\pi\in \Pi^\textrm{MD},\\
		\sum\limits_{s^\prime\in\mathcal{S}}\sum\limits_{a\in\mathcal{A}}c(s, a, s^\prime)P(s^\prime|s, a)\pi(a|s), & \textrm{if}~\pi\in \Pi^\textrm{MR}.
	\end{cases}
\end{align}
Similarly, the expected transmission cost in state $s$ is
\begin{align}
	\bm{f}^\pi(s)
	=\begin{cases}
		f(s, \pi(s)), & \textrm{if}~\pi\in \Pi^\textrm{MD},\\
		\sum\limits_{a\in\mathcal{A}}f(s, a)\pi(a|s), & \textrm{if}~\pi\in \Pi^\textrm{MR}.
	\end{cases}
\end{align}
Then, the average costs can be rewritten as
\begin{align}
	C_s(\pi) &= \sum_{s^\prime\in\mathcal{S}}\mathbf{P}_{s,s^\prime}^*(\pi) \bm{c}^\pi(s^\prime), \\
    F_s(\pi) &= \sum_{s^\prime\in\mathcal{S}}\mathbf{P}_{s,s^\prime}^*(\pi) \bm{f}^\pi(s^\prime),
\end{align}
where $\mathbf{P}^*(\pi) = \lim_{T\rightarrow\infty}\frac{1}{T}\sum_{t=1}^{T}\mathbf{P}^{t-1}(\pi)$ is the limiting matrix\footnote{Note that we replace $\limsup$ with $\lim$ since they converge to the same value for stationary policies\cite[Proposition~8.1.1]{puterman1994markov}.}. 

In average cost problems, the limiting behavior of $\mathbf{P}(\pi)$ is fundamental. When $\mathbf{P}(\pi)$ is \textit{unichain}, that is, it has a single recurrent class $\mathcal{R}$ and a set $\mathcal{T}$ of transient state, the limiting matrix $\mathbf{P}^*(\pi)$ has equal rows and satisfies\cite[Theorem~A.4]{puterman1994markov}
\begin{align}
	\mathbf{P}^*_{s,s^\prime}(\pi) = \begin{cases}
		\bm{\nu}(s^\prime), & s^\prime\in\mathcal{R},\\
		0, & s^\prime\in\mathcal{T},
	\end{cases}
\end{align}
where $\bm{\nu}$ is the stationary distribution of $\mathcal{R}$. Consequently, the average CAE and the average transmission cost are \textit{independent of initial states}, i.e., $C(\pi) = C_s(\pi), F(\pi) = F_s(\pi)$ for all $s\in\mathcal{S}$. When $\mathcal{T} = \emptyset$, then $\mathbf{P}(\pi)$ is \textit{recurrent}, and the average costs satisfy $C(\pi) = \bm{\nu}^\textrm{T}\bm{c}^\pi$ and $F(\pi) = \bm{\nu}^\textrm{T}\bm{f}^\pi$.

\subsection{Source-Agnostic Policy}
Before discussing the constrained optimal policy, we first examine a special subclass of stationary randomized policies, termed \textit{source-agnostic} (SA) policy. At each time $t$, the SA policy takes probabilistic actions independent of source states. Specifically, it selects each source $m$ with a fixed probability $f_m$, where $\sum_{m \in \mathcal{M}} f_m \leq F_{\max}$. The transition matrix of each subsystem $m$ of an SA policy $\pi_\textrm{sa}$ is given by
\begin{align}
	&\mathbf{P}^m_{s(m),s^\prime(m)}(\pi_\textrm{sa})= \notag\\
	&\begin{cases}
		\mathbf{Q}^m_{i,k} f_m p_s, &\textrm{if}~s(m)= (i,j), s^\prime(m) =  (k, i), \\
		\mathbf{Q}^m_{i,k} (1 - f_m p_s), &\textrm{if}~s(m) = (i,j),s^\prime(m) =  (k, j), \\
		\mathbf{Q}^m_{i,k}, &\textrm{if}~s(m) = (i,i),s^\prime(m) =  (k, i), \\
		0, &\textrm{otherwise},
	\end{cases}
\end{align}
where $i,j,k\in\mathcal{X}^m, i\neq j$. The system's transition matrix is 
\begin{align}
	\mathbf{P}_{s,s^\prime}(\pi_\textrm{sa}) = \prod_{m\in\mathcal{M}}\mathbf{P}^m_{s(m),s^\prime(m)}(\pi_\textrm{sa}).
\end{align}

The following result shows that the Markov chain induced by an SA policy is recurrent and aperiodic. This result is useful in proving the theorems in Section~\ref{sec:main_results} and Section~\ref{sec:lyapunov}.

\begin{proposition}\label{proposition:SA_policy} 
Suppose $\mathbf{Q}^m$ is irreducible and $0<f_m<F_{\max}$. Then $\mathbf{P}^m(\pi_\textrm{sa})$ forms a recurrent and aperiodic chain. 
\end{proposition}

\begin{proof}
    Since the self-transition probability of each state in $\mathcal{S}^m$ is positive, i.e., $\mathbf{P}^m_{s, s}(\pi_\textrm{sa})>0, \forall s \in \mathcal{S}^m$, then $\mathbf{P}^m(\pi_\textrm{sa})$ is aperiodic. To prove the recurrence of $\mathbf{P}^m(\pi_\textrm{sa})$, we need to show that every state $s\in\mathcal{S}^m$ communicates with every other state $s^\prime\in\mathcal{S}^m\backslash\{s\}$\cite[Chapter~4]{gallager1997discrete}. Suppose the chain is at state $s = (i,j)$, it can change to state $s^\prime = (k, j)$ with a positive probability $\mathbf{Q}^m_{i,k}(1-f_mp_s)$. Then, the chain can access state $s^{\prime\prime} = (l,k)$ with a probability of $\mathbf{Q}^m_{k,l}f_mp_s$. So, $s^{\prime\prime}$ is accessible from $s$. Likewise, we can show that $s$ is accessible from $s^{\prime\prime}$. Hence, $s$ communicates with $s^{\prime\prime}$. This holds for any $s, s^{\prime\prime} \in \mathcal{S}^m$, which completes the proof.
\end{proof}

\begin{corollary}\label{corollary:recurrent}
	It follows from Proposition~\ref{proposition:SA_policy} that $\mathbf{P}(\pi_\textrm{sa})$ is recurrent and aperiodic. 
\end{corollary}

\section{Main Results}\label{sec:main_results}
In this section, we show the existence and structure of a constrained optimal policy. We first relax the original CMDP as an unconstrained Lagrangian MDP with a multiplier $\lambda\geq0$, termed $\mathscr{P}_\textrm{L-MDP}^\lambda$. Theorem~\ref{theorem:lambda_optimal_policy} shows that the Lagrangian MDP for any given $\lambda\geq0$ is \textit{communicating} and the corresponding $\lambda$-optimal policy is \textit{stationary deterministic}. We use \textit{relative value iteration} to solve $\mathscr{P}_\textrm{L-MDP}^\lambda$ and show its convergence in Proposition~\ref{proposition:convergence_of_rvi}. Proposition~\ref{proposition:piecewise_continuity} shows some important properties of the cost functions, based on which we establish the structure of an optimal policy in Theorem~\ref{theorem:structure_of_optimal_policy}.

\subsection{The Lagrangian MDP}\label{sec:lagrangian-mdp}
Now we derive the structure of an optimal policy based on the Lagrangian technique\cite{altman1999constrained}. We start with rewriting $\mathscr{P}_\textrm{CMDP}$ in its  Lagrangian form. The \textit{Lagrangian average cost} of a policy $\pi$ with multiplier $\lambda\geq 0$ is defined as
\begin{align}
	\mathcal{L}_s^\lambda(\pi)
	&\triangleq \limsup_{T\rightarrow\infty}\frac{1}{T}\sum_{t=1}^{T}\mathbb{E}^\pi \big[l^\lambda(S_t, A_t, S_{t+1})|S_1 = s\big]\notag\\
	&= C_s(\pi) + \lambda F_s(\pi),
\end{align}
where the \textit{Lagrangian instantaneous cost}, $l^\lambda(s, a, s^\prime)$, is
\begin{align}
	l^\lambda(s, a, s^\prime) = c(s, a, s^\prime) + \lambda f(s, a).
\end{align}
Given any $\lambda\geq0$, define the unconstrained \textit{Lagrangian MDP}
\begin{align}
	\mathscr{P}_\textrm{L-MDP}^\lambda:\,\, \inf_{\pi\in \Pi}\mathcal{L}_s^\lambda(\pi).
\end{align}
The above problem is a standard MDP with an average cost criterion. We note that the case $\lambda = 0$ corresponds to the trivial case where the constraint in \eqref{problem:CMDP} is not active. 

The Lagrangian approach thereafter solves the following two-layer unconstrained ``sup-inf" problem\cite{altman1999constrained}:
\begin{align}
    \sup_{\lambda\geq 0}\,\inf_{\pi\in \Pi}\,\mathcal{L}_s^\lambda(\pi) - \lambda F_{\max}.\label{problem:lagrangian}
\end{align}
The inner layer solves a standard MDP with a fixed multiplier, while the outer layer searches for the largest multiplier that produces a feasible policy. In the rest of the paper, we shall refer to $\mathscr{P}^\lambda_\textrm{L-MDP}$ as the inner problem since $\lambda F_{\max}$ is constant.

In what follows we show that for any $\lambda\geq0$ and $F_{\max}\in(0, 1]$, there exists an unconstrained optimal policy for the inner problem, that is stationary and Markovian deterministic. We show in Section~\ref{sec:structure} the structural results of Lagrangian MDPs, based on which the supremum of the outer problem can be easily achieved. It is known that the \textit{duality gap}, defined as the difference between the optimal values obtained in problems \eqref{problem:CMDP} and \eqref{problem:lagrangian}, is generally non-zero. An important consequence of Theorem~\ref{theorem:structure_of_optimal_policy} is that our approach leads to a constrained optimal policy, i.e., achieving zero duality gap.
 
\begin{definition} ($\lambda$-optimal policy)
    A policy $\pi$ is called $\lambda$-optimal if it solves $\mathscr{P}^\lambda_\textrm{L-MDP}$ for a given multiplier $\lambda$ and the average costs $\mathcal{L}^\lambda, C^\lambda$, and $F^\lambda$ are independent of initial states. 
\end{definition}

\begin{theorem}\label{theorem:lambda_optimal_policy}
For any given $\lambda\geq 0$, there exists a vector $\bm{h}^\lambda$ that satisfies the following Bellman's optimality equation:
\begin{align}
    \mathcal{L}^\lambda + \bm{h}^\lambda(s) = \min_{a\in \mathcal{A}}\sum_{s^\prime \in \mathcal{S}}P(s^\prime|s,a)\big(l^\lambda(s,a, s^\prime)+\bm{h}^\lambda(s^\prime)\big),\label{problem:bellman}
\end{align}
for all $s\in\mathcal{S}$. The quantities $\mathcal{L}^\lambda, F^\lambda$, and $C^\lambda$ are independent of initial states. Furthermore, the $\lambda$-optimal policy belongs to $\Pi^\textrm{MD}$ and can be obtained by
\begin{align}
    \pi^*_\lambda(s)&=\argmin_{a\in \mathcal{A}}\sum_{s^\prime \in \mathcal{S}}P(s^\prime|s,a)\big(l^\lambda(s,a, s^\prime)+\bm{h}^\lambda(s^\prime)\big).\label{eq:bellman_pi}
\end{align}
\end{theorem}
\begin{proof}
    According to \cite[Theorem~8.3.2]{puterman1994markov} and \cite[Chapter~4.2]{bertsekas2011dynamic}, it is sufficient to show that the transition kernel $P$ is (weakly) communicating. By \cite[Proposition~8.3.1]{puterman1994markov}, this holds if there exists a randomized stationary policy that induces a recurrent chain. Clearly, the SA policy satisfies this condition (see Proposition~\ref{proposition:SA_policy}), which completes the proof. 
\end{proof}

We can use the \textit{relative value iteration} (RVI) \cite[Section~8.5.5]{puterman1994markov} to solve $\mathscr{P}_\textrm{L-MDP}^\lambda$ for any given $\lambda\geq 0$, as shown in Algorithm~\ref{alg:RVI_LMDP}. Let $s_\textrm{ref}\in\mathcal{S}$ denote an arbitrarily chosen reference state, $\bm{v}^k$ and $\tilde{\bm{v}}^k$ denote the value function and the relative value function for iteration $k$, respectively. In the beginning, we set $\bm{v}^{0}= \tilde{\bm{v}}^{0}= \bm{0}$. Then, for each $k\geq 1$, the RVI updates the value functions as follows:
\begin{align}
	\bm{v}^{k+1}(s) &=  \min_{a\in \mathcal{A}}\sum_{s^\prime \in \mathcal{S}}P(s^\prime|s,a)\big(l^\lambda(s,a, s^\prime)+\tilde{\bm{v}}^{k}(s^\prime)\big), \label{eq:RVI-a}\\
	\tilde{\bm{v}}^{k+1}(s) &= \bm{v}^{k+1}(s) - \bm{v}^{k+1}(s_{\textrm{ref}}).\label{eq:RVI-b}
\end{align}

As stated in Proposition~\ref{proposition:convergence_of_rvi} below, the sequences $\{\bm{v}^k\}_{k\geq 0}$ and $\{\tilde{\bm{v}}^k\}_{k\geq 0}$ converge to $\bm{v}^*$ and $\bm{h}^\lambda$, respectively. Then, the $\lambda$-optimal policy can be obtained by solving \eqref{eq:bellman_pi}, and the average costs are obtained as
\begin{align}
	C^\lambda &= \sum_{s\in\mathcal{S}}\bm{\nu}^{\pi_\lambda^*}(s)\bm{c}^{\pi_\lambda^*}(s),\label{eq:average-cost-C}\\
	F^\lambda &= \sum_{s\in\mathcal{S}}\bm{\nu}^{\pi_\lambda^*}(s)\bm{f}^{\pi_\lambda^*}(s), \label{eq:average-cost-F}\\
    \mathcal{L}^\lambda &=  C^\lambda + \lambda F^\lambda, \label{eq:average-cost-L}
\end{align}
where $\bm{\nu}^{\pi_\lambda^*}$ is the stationary distribution of $\mathbf{P}(\pi_\lambda^*)$, and $\bm{\nu}^{\pi_\lambda^*}(s)$ can be interpreted as the fraction of time the system is in $s$.

\begin{proposition}\label{proposition:convergence_of_rvi}
For any $\lambda\geq 0$, the sequences $\{\bm{v}^k\}_{k\geq 0}$ generated by Eqs.~\eqref{eq:RVI-a}-\eqref{eq:RVI-b} convergences, i.e., $\lim_{k\rightarrow\infty} \bm{v}^k \rightarrow \bm{v}^*$. Moreover, $\mathcal{L}^\lambda = \bm{v}^*(s_\textrm{ref})$ and $\bm{h}^\lambda= \bm{v}^* - \bm{v}^*(s_\textrm{ref})$ is a solution to \eqref{problem:bellman}.
\end{proposition}
\begin{proof}
    A sufficient condition for the convergence of the RVI is the recurrence and aperiodicity of the Markov chain induced by optimal stationary policies\cite[Proposition~4.3.2]{bertsekas2011dynamic},\cite[Proposition~6.6.3]{sennott1998stochastic}. The recurrence of the optimal policy follows directly from Theorem~\ref{theorem:lambda_optimal_policy}. Moreover, the self-transition probabilities satisfy $P(s|s,a)>0$ for any $s\in\mathcal{S}$ and $a\in\mathcal{A}$. Therefore, aperiodicity can always be achieved.
\end{proof}

\begin{algorithm}[ht]
    \renewcommand{\thealgocf}{1}
    \DontPrintSemicolon
    \setstretch{0.8}
    \SetAlgoLined
	\KwIn{$\lambda$, $c$, $f$, $P$, $\epsilon$}
	\KwOut{$\pi_\lambda^*, F^\lambda, C^\lambda$, $\mathcal{L}^\lambda$}
	\BlankLine
	$s_\textrm{ref}=0, \bm{v}^0 = \tilde{\bm{v}}^0=\bm{0}$ for all $s\in\mathcal{S}$\;
	\For{$k = 1, 2, \ldots$}{ 
		\For{each $s\in \mathcal{S}$}{   
		  \textsc{Update} $\bm{v}^k, \tilde{\bm{v}}^k$ using~\eqref{eq:RVI-a}-\eqref{eq:RVI-b}.\;
		}
		\If{$\|\bm{v}^{k}-\bm{v}^{k-1}\|_{\infty} < \epsilon$}{
            \For{each $s\in \mathcal{S}$}{    
                \textsc{Compute} $\pi_\lambda^*$ using~\eqref{eq:bellman_pi} with $\bm{h}^\lambda = \tilde{\bm{v}}^k$.\;
            }
            \textsc{Compute} average costs using~\eqref{eq:average-cost-C}-\eqref{eq:average-cost-L}.\;
            \Return{$\pi_\lambda^*, F^\lambda, C^\lambda$, $\mathcal{L}^\lambda$}
		}
	}
	\caption{RVI for solving $\mathscr{P}_\textrm{L-MDP}^\lambda$.}
 \label{alg:RVI_LMDP}
\end{algorithm}

\subsection{Structure of the Optimal Policy}\label{sec:structure}
Until now, we have shown the existence of an optimal deterministic policy for the inner problem. In the following, we show that there exists a constrained optimal policy that is in no case more complex than a randomized mixture of two deterministic policies. Before stating that, we first present some important properties of the cost functions $\mathcal{L}^\lambda$, $C^\lambda$, and $F^\lambda$. These properties provide the basis for solving the outer problem in \eqref{problem:lagrangian}.

\begin{proposition} \label{proposition:piecewise_continuity}
$\mathcal{L}^\lambda$ is a piecewise linear, concave, and monotonically increasing function of $\lambda$. $F^\lambda (C^\lambda)$ is piece-wise constant, monotonically non-increasing (non-decreasing).
\end{proposition}
\begin{proof}
    The monotonicity of $\mathcal{L}^\lambda$ and $F^\lambda$ is demonstrated in \cite[Lemma~3.1]{beutler1985optimal} and \cite[Lemma~3.5]{sennott_1993}. It remains to show that $\mathcal{L}^\lambda$ is piecewise-linear and concave (PWLC) and $F^\lambda (C^\lambda)$ is piecewise constant. We give a proof by induction based on the recursion \eqref{eq:RVI-a}-\eqref{eq:RVI-b}. Since $l^\lambda(s,a,s^\prime) = c(s,a,s^\prime) + \lambda f(s,a)$ is a linear function, then for all $s\in\mathcal{S}$ $\bm{v}^1(s) = \min_{a\in \mathcal{A}}\sum_{s^\prime \in \mathcal{S}}P(s^\prime|s,a)l^\lambda(s,a,s^\prime)$ is a PWLC function of $\lambda$. Then, the function $\tilde{\bm{v}}^1(s) = \bm{v}^1(s) - \bm{v}^1(s_\textrm{ref})$ is also PWLC. Assume the following induction hypothesis: $\tilde{\bm{v}}^k(s)$ is PWLC for all $s\in\mathcal{S}$. Consequently, $\sum_{s^\prime\in\mathcal{S}}P(s^\prime|s,a)\tilde{\bm{v}}^k(s^\prime)$ is PWLC since the summation of PWLC functions is PWLC. It follows that $\bm{v}^{k+1}(s) =  \min_{a\in \mathcal{A}}\sum_{s^\prime\in \mathcal{S}}P(s^\prime|s,a)(l^\lambda(s,a,s^\prime) + \tilde{\bm{v}}^k(s^\prime))$ is still PWLC since the $\min$ of PWLC functions is PWLC. By Proposition~\ref{proposition:convergence_of_rvi}, $\bm{v}^k$ converges to a vector $\bm{v}^*$. Consequently, $\mathcal{L}^\lambda = \bm{v}^*(s_\textrm{ref})$ is PWLC. $F^\lambda$ is the derivative of $\mathcal{L}^\lambda$, so it is piecewise constant. 
\end{proof}

\begin{definition}
($\gamma$ value) By the monotonicity of $F^\lambda$ and $\mathcal{L}^\lambda$, the solution to the outer problem in \eqref{problem:lagrangian} can be defined by
\begin{align}
    \gamma \triangleq \inf\{\lambda:\lambda\geq 0, F^\lambda \leq F_{\max}\}.
\end{align}
\end{definition}

\begin{figure*}[t]
    \centering
    \includegraphics[width=0.8\linewidth]{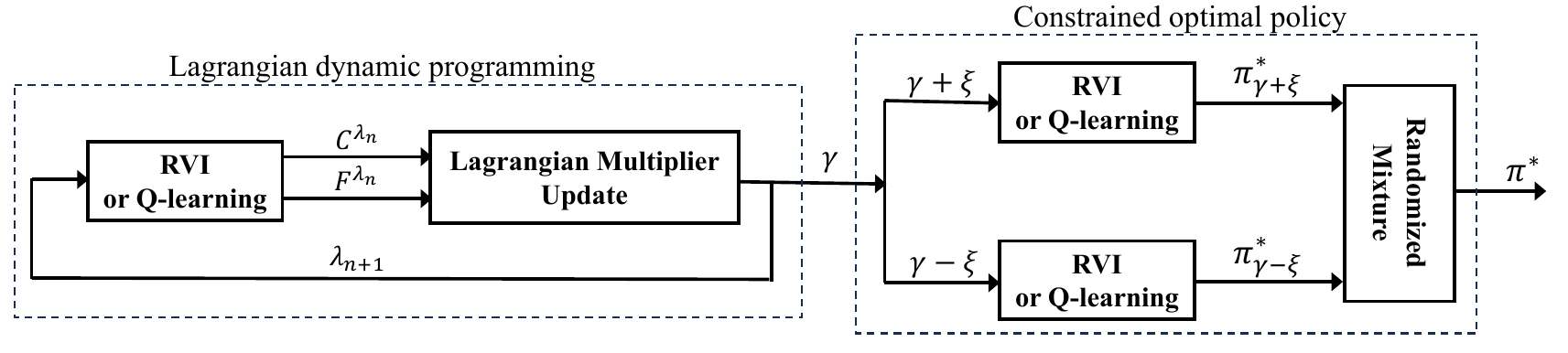}
    \vspace{-0.1in}
    \caption{Schematic representation of the flow diagram of finding the constrained optimal policy of CMDPs. The RVI algorithm is applied to solve Lagrangian MDPs with a fixed $\lambda$. We will develop a Q-learning algorithm for unknown environments in Section~\ref{sec:drl}. The Lagrangian multiplier is updated until the $\gamma$ value is found. Then, the constrained optimal policy is constructed as a mixture of two deterministic policies. The pseudo-code is summarized in Algorithm~\ref{alg:Sa-RVI}.}
    \label{fig:lagrangian_approach}
\end{figure*}

\begin{proposition} \label{proposition:gamma}
For any $F_{\max}\in(0, 1]$, the unconstrained problem in \eqref{problem:lagrangian} always returns a solution, that is $0\leq \gamma<\infty$.
\end{proposition}
\begin{proof}
    Suppose that the assertion of the Lemma is false. Then $\gamma = \infty$ and $F^\lambda > F_{\max}$ for all $\lambda \geq 0$. Consequently, $\mathcal{L}^\lambda = C^\lambda + \lambda F^\lambda > C^\lambda + \lambda F_{\max}$. On the other hand, we can construct an SA policy $\pi_\textrm{sa}$ that satisfies $F_s(\pi_\textrm{sa}) = F_{\max} - \sigma$, where $0\leq \sigma<F_{\max}$. This can be achieved by setting $\sum_{m \in \mathcal{M}}f_m = F_{\max} - \sigma$. For this $\pi_\textrm{sa}$, we have $\mathcal{L}^\lambda_s(\pi_\textrm{sa}) = C_s(\pi_\textrm{sa})+ \lambda F_s(\pi_\textrm{sa}) = C_s(\pi_\textrm{sa}) + \lambda(F_{\max}- \sigma)$. Hence, $\mathcal{L}^\lambda>\mathcal{L}^\lambda_s(\pi_\textrm{sa})$ for sufficient large $\lambda$, which is clearly a contradiction as $\mathcal{L}^\lambda$ attains the infimum. The above holds for any $0<F_{\max}\leq 1$, which completes the proof.
\end{proof}

Although the above propositions ensure the existence of an optimal solution to problem \eqref{problem:lagrangian}, it may not be optimal to the original CMDP \eqref{problem:CMDP} since the duality gap can be non-zero. The authors in \cite[Theorem~4.3]{beutler1985optimal} provide the sufficient conditions under which an optimal policy $\pi^*$ for a Lagrangian MDP is also optimal for the original constrained problem $\mathscr{P}_\textrm{CMDP}$, i.e., zero duality gap. Specifically, the policy $\pi^*$ must solve $\mathscr{P}^\lambda_\textrm{L-MDP}$ for some $\lambda> 0$, and in addition, the average costs satisfy $F^\lambda = F_s(\pi^*) = F_{\max}, C^\lambda = C_s(\pi^*)$ for all $s\in \mathcal{S}$. Notice that $\pi^*$ needs not to be deterministic. We verify these conditions and establish the structure of the optimal policy in the following result.

\begin{theorem} 
\label{theorem:structure_of_optimal_policy}
The optimal policy $\pi^*$ belongs to $\Pi^\textrm{MD}$ if the following condition holds:
\begin{align}
    \gamma (F^\gamma - F_{\max}) = 0.
\end{align}
Otherwise, the optimal policy is a mixture of two policies in space $\Pi^\textrm{MD}$, i.e.,
\begin{align}
    \pi^* = q \pi_{\gamma^{-}}^* + (1-q) \pi_{\gamma^{+}}^*,\label{eq:prove_mixture}
\end{align}
where $q = (F_{\max} - F^{\gamma^{+}})/(F^{\gamma^{-}} - F^{\gamma^{+}})$, $F^{\gamma^{-}}$ and $F^{\gamma^{+}}$ are the left-hand limit and right-hand limit at point $\gamma$. Notice that $\pi_{\gamma^{-}}^*$ is infeasible while $\pi_{\gamma^{+}}^*$ is feasible\footnote{Recall that $F^\lambda$ and $C^\lambda$ may have a jump discontinuity at point $\gamma$. $\pi_{\gamma^{-}}^*$ and $\pi_{\gamma^{+}}^*$ can behave quite differently. Hence, we cannot simply apply the feasible policy $\pi_{\gamma^{+}}^*$ since it may result in undesirable performance.}.
\end{theorem}

\begin{proof}
   Since $F^\lambda$ is a piecewise constant and non-increasing function of $\lambda$ (see Proposition~\ref{proposition:piecewise_continuity}), it can be expressed as
    \begin{align}
        F^\lambda = \sum_{i = 0}^{N}F^{\lambda_i}\mathds{1}(\lambda\in[\lambda_i, \lambda_{i+1})),
    \end{align}
	where $0 = \lambda_0 < \lambda_1 < \ldots <\lambda_N <\lambda_{N+1} \rightarrow \infty$ are the breakpoints of $F^\lambda$, $F^0> F^{\lambda_1} > \ldots > F^{\lambda_N}$ are the possible values of $F^\lambda$. Since $\gamma<\infty$ for any $F_{\max} \in (0, 1]$ (see Proposition~\ref{proposition:gamma}), we prove the theorem by distinguishing between three cases:

	(1) When $F^0 \leq F_{\max}$, it implies that $\gamma = 0$ and the unconstrained policy $\pi_0^*$ is feasible. By \cite[Lemma~3.1]{beutler1985optimal}, $C^\lambda$ is monotonically non-decreasing in $\lambda$. Hence, $\pi_0^*$ achieves the minimum average CAE, $C^0$, and it is therefore optimal. 
	
	(2) When $F^{\lambda_n} = F_{\max}$ for some $n$, then any $\lambda\in [\lambda_n, \infty)$ can result in a feasible $\lambda$-optimal policy. So we have $\gamma = \lambda_n$. The corresponding $\gamma$-optimal policy $\pi_\gamma^*$ satisfies $F(\pi_\gamma^*) = F_{\max}$ and $\mathcal{L}(\pi_\gamma^*) = \mathcal{L}^\gamma$. Therefore, $\pi_\gamma^*$ is constrained optimal.
	
	(3) When $F^{\lambda_{n-1}} < F_{\max} < F^{\lambda_{n}}$ for some $n$, then $[\lambda_{n-1}, \lambda_{n})$ is an infeasible region and $[\lambda_n, \infty)$ is a feasible region. So we have $\gamma = \lambda_n$. For this case, it is proved in \cite[Theorem~4.4]{beutler1985optimal} that the policy in \eqref{eq:prove_mixture} satisfies the sufficient conditions and it is therefore constrained optimal.
\end{proof}

\begin{corollary}\label{corollary:cusp}
$\gamma$ is a corner of $\mathcal{L}^\lambda$ and a breakpoint of $F^\lambda$.
\end{corollary}
\begin{proof}
This follows from the proof of Theorem~\ref{theorem:structure_of_optimal_policy}.
\end{proof}

\section{Computation of the Optimal Policy}\label{sec:solutions}
This section presents algorithms for finding the optimal policy both with and without \emph{priori} knowledge of the source and channel statistics. When the system statistics are known a \emph{priori}, we propose a novel policy search algorithm, termed \emph{intersection search}, which exploits the monotonicity of $\mathcal{L}^\lambda$ to reduce computation overhead. Moreover, we propose a Reinforcement Learning (RL) algorithm that asymptotically achieves the minimum in unknown environments.

\subsection{An Efficient Search Method for CMDP}\label{sec:intersection_search} 
\subsubsection{Bisection search} Fig.~\ref{fig:lagrangian_approach} illustrates the standard procedure for solving CMDPs\cite{rcpo,djonin2007Q}. Specifically, the inner loop seeks a $\lambda$-optimal policy for a given $\lambda$, while in the outer loop, $\lambda$ is updated until the optimal multiplier $\gamma$ is found. However, in practice, computing $\gamma$ and the corresponding optimal policies is computationally demanding\cite{Ma1986}.

A commonly used method is the \textit{Bisection search plus RVI} (Bisec-RVI) \cite{bisection_robots, AoII_GC}. It leverages the monotonicity of $F^\lambda$ to determine $\gamma$. More precisely, the ``root'' of $F^\lambda - F_{\max}$ in a predetermined interval $[0, \lambda_{\max}]$, where $\lambda_{\max}$ is a reasonably large constant satisfying $F^{\lambda_{\max}} <F_{\max}$. Notice that the ``root'' needs not to exist because of the discontinuity of $F^\lambda$. However, we can bracket $\gamma$ within an arbitrarily small interval. 

At each iteration, the Bisec-RVI method divides the interval into two halves and selects the subinterval that brackets the root. We denote the interval at the $n$-th iteration as $I_n = [\lambda_{-}^{n}, \lambda_{+}^{n}]$ and the middle point of $I_n$ as $\gamma^n$, where $I_0 = [0, \lambda_{\max}]$ and $\gamma^n = (\lambda_{-}^{n} +\lambda_{+}^{n})/2$. The Bisec-RVI uses the RVI to solve the problem $\mathscr{P}_\textrm{L-MDP}^{\gamma^n}$. If the average transmission cost at point $\gamma^n$ is larger than $F_{\max}$, it implies that the root $\gamma$ is located within the subinterval $[\gamma^n, \lambda_{+}^{n}]$. Therefore, we can update the interval as $I_{n+1} = [\gamma^n, \lambda_{+}^{n}]$. On the other hand, if the average transmission cost at point $\gamma^n$ is smaller than $F_{\max}$, then we have $I_{n+1} = [ \lambda_{-}^{n}, \gamma^n]$. This process continues until a desired accuracy is achieved, i.e., $|\lambda_{-}^n - \lambda_{+}^n| < \xi$. Then, the optimal policy is given by
\begin{align}
	\pi^* = \mu \pi^*_{\lambda^n_{-}} + (1-\mu)\pi^*_{\lambda^n_{+}},\label{eq:optimal_policy}
\end{align}
where $\mu$ is a randomization factor given by
\begin{align}
	\mu = \frac{F_{\max} - F^{\lambda^n_{+}}}{F^{\lambda^n_{-}} - F^{\lambda^n_{+}}}.\label{eq:randomization_factor}
\end{align}

\begin{remark}
    The time complexity of the bisection search is $\mathcal{O}(\log_2\frac{\lambda_{\max}}{\xi})$. So it requires a large number of iterations to reach a desired accuracy. Moreover, it exhibits a slow convergence rate near the breakpoint $\gamma$ because $\frac{dF^\lambda}{d\lambda}$ is zero.
\end{remark}
\begin{figure}[t]
    \centering
    \includegraphics[width=0.9\linewidth]{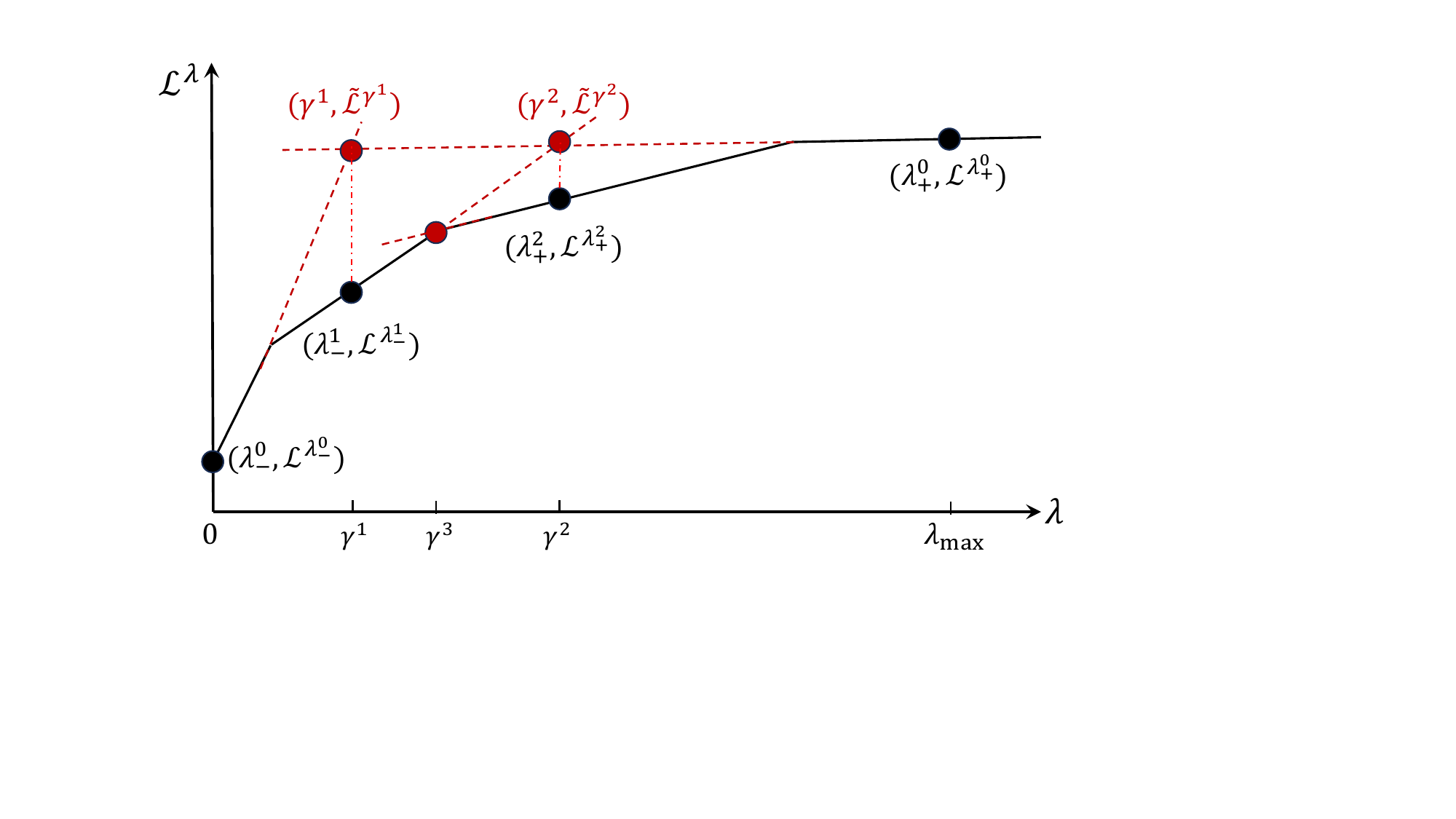}
    \caption{An illustration of the Insec-RVI method. Consider that function $\mathcal{L}^\lambda$ has $4$ segments within the interval $[0, \lambda_{\max}]$ and $\gamma$ is the second corner of $\mathcal{L}^\lambda$. The red dots represent the intersection points at each iteration, while the black dots are the corresponding projected points onto the curve $\mathcal{L}^\lambda$.}
    \label{fig:sarvi}
\end{figure}
\subsubsection{Intersection search} Next, we propose a novel method, called \textit{Intersection search plus RVI} (Insec-RVI), that efficiently solves $\mathscr{P}_\textrm{CMDP}$ by exploiting the concavity and piecewise linearity of $\mathcal{L}^\lambda$. In this method, $\gamma^n$ is no longer the middle point of $I_n$ but an intersection point of two tangents to the function $\mathcal{L}^\lambda$, as shown in Fig~\ref{fig:sarvi}. More specifically, one tangent is formed by a point $(\lambda_{-}^n, \mathcal{L}^{\lambda_{-}^n})$ with a slop of $F^{\lambda_{-}^n}$, and another tangent is formed by a point $(\lambda_{+}^n, \mathcal{L}^{\lambda_{+}^n})$ with a slop of $F^{\lambda_{+}^n}$. Then, the intersection point of these two lines, $(\gamma^n, \tilde{\mathcal{L}}^{\gamma^n})$, is given by
\begin{align}
	\gamma^n 
	&= \frac{C^{\lambda_{+}^n}-C^{\lambda_{-}^n}}{F^{\lambda_{-}^n} - F^{\lambda_{+}^n}},\\
	\tilde{\mathcal{L}}^{\gamma^n} &= F^{\lambda_{-}^n}(\gamma^n - \lambda_{-}^n) + \mathcal{L}^{\lambda_{-}^n}.
\end{align}
Analogous to the Bisec-RVI method, if $F^{\gamma^n}\geq F_{\max}$, then $\lambda_{-}^{n+1} = \gamma^n$, otherwise $\lambda_{+}^{n+1} = \gamma^n$. 

According to Corollary~\ref{corollary:cusp}, $\gamma$ is a corner of $\mathcal{L}^\lambda$. Therefore, the algorithm converges when the intersection point is located on $\mathcal{L}^\lambda$, that is $\tilde{\mathcal{L}}^{\gamma^n} =\mathcal{L}^{\gamma^n}$. To compute the optimal policy, we can simply use Eqs.~\eqref{eq:optimal_policy}-\eqref{eq:randomization_factor} by setting $\lambda_{-}^n = \gamma^n - \zeta/4, \lambda_{+}^n = \gamma^n + \zeta/4$, where $\zeta$ is an arbitrarily chosen small perturbation. Recall that the stopping criterion of the Bisec-RVI is $|\lambda_{-}^n - \lambda_{+}^n|<\xi$. Therefore, the Insec-RVI can achieve comparable performance to the Bisec-RVI by setting $\zeta = \xi$, without introducing any additional computations. In addition, choosing an appropriate $\lambda_{\max}$ value can be problematic when using Bisec-RVI since it can either be too large, leading to excessive computation time, or too small, rendering it infeasible. Fortunately, in our Insec-RVI algorithm, we can simply initialize the costs as $F^{\lambda_{+}} \leftarrow 0$ and $C^{\lambda_{+}} \leftarrow C_{\max}$, where $C_{\max}$ is the average CAE induced by the non-transmission policy. The pseudo-code is summarized in Algorithm~\ref{alg:Sa-RVI}. 

\begin{remark}
    To the best of our knowledge, intersection search is the first method to solve single-constraint CMDPs by leveraging the piecewise linearity of $\mathcal{L}^\lambda$. Suppose $\mathcal{L}^\lambda$ has $W$ segments in the range $[0, \lambda_{\max}]$. The Insec-RVI method requires at most $W-1$ iterations, which is much fewer compared to the Bisec-RVI method as $W-1 \ll \log_2(\frac{\lambda_{\max}}{\xi})$ when $\xi\rightarrow0$. 
\end{remark}

\begin{algorithm}[htbp]
    \renewcommand{\thealgocf}{2}
    \DontPrintSemicolon
    \setstretch{0.8}
    \SetAlgoLined
    \algsetup{linenosize=\small}
\SetKwInOut{Input}{Input}
\SetKwInOut{Output}{Output}
\SetKwInOut{Initilize}{Initilize}
\KwIn{$c$, $f$,$P$, $\lambda_{\max}, \epsilon, \xi$}
\KwOut{$\gamma, \pi^*$}
\BlankLine
\tcc{check if $\gamma=0$}
$\pi_{\lambda_{-}}^*, F^{\lambda_{-}}, C^{\lambda_{-}}, \mathcal{L}^{\lambda_{-}} \leftarrow\operatorname{RVI}(0, f, d, P, \epsilon)$\;
\If{$F^{\lambda_{-}} \leq F_{\max}$}{
    \Return{$\gamma \leftarrow 0, \pi^* \leftarrow\pi_{\lambda_{-}}^*$}
}
\BlankLine
\tcc{find a $\gamma>0$ with $F^\gamma = F_{\max}$}
$F^{\lambda_{+}}\leftarrow0, C^{\lambda_{+}}\leftarrow C_{\max}, \mathcal{L}^{\lambda_{+}} \leftarrow C_{\max}$\;
\While{True}{
    \tcp{Compute the intersection point}
    $\gamma = (C^{\lambda_{+}}-C^{\lambda_{-}})/(F^{\lambda_{-}} - F^{\lambda_{+}})$\;
    $\tilde{\mathcal{L}}^\gamma = F^{\lambda_{-}}(\gamma - \lambda_{-}) + \mathcal{L}^{\lambda_{-}}$\;
    \tcp{Check the stopping condition}
    $F^\gamma, C^\gamma, \mathcal{L}^\gamma \leftarrow\rvi(\gamma, f, d, P, \epsilon)$\;
    \eIf{$\mathcal{L}^{\gamma} = \tilde{\mathcal{L}}^\gamma$}{
        \tcp{Compute the optimal policy}
        \eIf{$F^\gamma = F_{\max}$}{
            $\pi^* \leftarrow \pi^*_\gamma$\;
        }{
            $\pi_{\lambda_{-}}^*, F^{\lambda_{-}}\leftarrow\rvi(\gamma - \xi/4, f, d, P, \epsilon)$\;
            $\pi_{\lambda_{+}}^*, F^{\lambda_{+}}\leftarrow\rvi(\gamma + \xi/4, f, d, P, \epsilon)$\;
            $\mu\leftarrow (F_{\max} - F^{\lambda_{+}}) / (F^{\lambda_{-}} - F^{\lambda_{+}})$\;
            $\pi^* \leftarrow \mu \pi_{\lambda_{-}}^* + (1-\mu) \pi_{\lambda_{+}}^*$\;
        }
        \Return{$\gamma, \pi^*$}
    }{
        \tcp{Shrink the interval}
        \eIf{$F^{\gamma} \leq F_{\max}$}{
            $\lambda_{+} \leftarrow \gamma$, $C^\lambda_{+}\leftarrow C^\gamma$, $F^\lambda_{+}\leftarrow F^\gamma$, $\mathcal{L}^\lambda_{+}\leftarrow \mathcal{L}^\gamma$\;
        }{
            $\lambda_{-} \leftarrow \gamma$, $C^\lambda_{-}\leftarrow C^\gamma$, $F^\lambda_{-}\leftarrow F^\gamma$, $\mathcal{L}^\lambda_{-}\leftarrow \mathcal{L}^\gamma$\;
        }
    }
}
\caption{Insec-RVI for solving $\mathscr{P}_\textrm{CMDP}$.}
 \label{alg:Sa-RVI}
\end{algorithm} 
\vspace{-0.1in}
\subsection{RL for Unknown Environments}\label{sec:drl}
Until now, we have focused on cases where the system transition probabilities are known a \textit{priori}. We rely on this information to compute expected costs in the RVI recursions. However, in practice, it is common for this information to be unknown at the time of deployment or to change over time. 

We propose an RL algorithm to address this challenge. The learning algorithm follows the procedure\footnote{A more efficient learning procedure is to alternate RVI/Q-learning updates with a Lagrangian multiplier update in a single loop\cite{PDO2018, lima2022model}. However, the optimality of this method has not been formally established \cite{rcpo}} in Fig.~\ref{fig:lagrangian_approach}. Specifically, we replace the original RVI recursion \eqref{eq:RVI-a}-\eqref{eq:RVI-b} with a model-free (free of transition probabilities) iterates\cite{rviQlearning, djonin2007Q}. The RL agent learns optimal policies through trial and error by observing state-action-reward transitions.

Now, we present the essence of the learning algorithm. Define the Q-factor $\bm{q}^\lambda(s,a)$ so that
\begin{align}
	\bm{h}^\lambda(s)&= \min_{a\in \mathcal{A}}\bm{q}^\lambda(s,a).
\end{align}
Then, the Bellman's equation \eqref{problem:bellman} can be written as
\begin{align}
	\mathcal{L}^\lambda \hspace{-0.4em}+ \bm{q}^\lambda(s, a) \hspace{-0.2em}= \hspace{-0.4em}
 \sum_{s^\prime\in \mathcal{S}}P(s^\prime|s, a)(l^\lambda(s, a,s^\prime)
	\hspace{-0.2em}+\hspace{-0.2em}\min_{a^\prime\in \mathcal{A}}\bm{q}^\lambda(s^\prime, a^\prime)).
\end{align}
The RVI for the Q-factors is given by
\begin{align}
	\bm{q}^{k+1}&(s,a) + \min_{a^\prime\in \mathcal{A}}\bm{q}^k(s_\textrm{ref},a^\prime) = \notag\\
    &\sum_{s^\prime\in \mathcal{S}}P(s^\prime|s,a)\big(l^\lambda(s,a,s^\prime) + \min_{a^\prime\in \mathcal{A}}\bm{q}^k(s^\prime, a^\prime)\big). \label{eq:q-iterator}
\end{align}
The sequence $\{\min_{a\in \mathcal{A}}\bm{q}^k(s,a)\}_{k\geq 0}$ is expected to converge to $\bm{h}^\lambda(s)$, and the sequence $\{\min_{a\in \mathcal{A}}\bm{q}^k(s_\textrm{ref},a)\}_{k\geq0}$ is expected to converge to $\mathcal{L}^\lambda$.

Then the Q-factor can be estimated via the model-free \textit{average-cost Q-learning}\cite{rviQlearning} by the recursion
\begin{align}
	\bm{q}^{k+1}(s,a) &= (1-\alpha_k)\bm{q}^k(s,a) + \alpha_k
	\big(
		l^\lambda(s,a,s^\prime)  + \notag\\
		 &\qquad\quad \min_{a^\prime\in \mathcal{A}}\bm{q}^k(s^\prime,a^\prime) 
		 - \min_{a^\prime\in \mathcal{A}}\bm{q}^k(s_\textrm{ref},a^\prime)
	\big),\label{eq:q-learning}
\end{align}
where $s^\prime$ is generated from the pair $(s,a)$ by simulation, $\alpha_k\in(0,1)$ is the learning rate that may be constant or diminishing over time. The sequence $\{\bm{q}^k\}$ remains bounded and convergences to $\bm{q}^*$ almost surely\cite[Theorem~3.5]{rviQlearning}. The Q-learning algorithm for solving $\mathscr{P}_\textrm{L-MDP}^\lambda$ is presented in Algorithm~\ref{alg:QL_LMDP}. To solve $\mathscr{P}_\textrm{CMDP}$, one simply needs to replace the RVI recursions in Algorithm~\ref{alg:Sa-RVI} with the Q-learning recursions.

However, the use of value iteration in Insec-RVI and RL leads to significant convergence times and increased computational complexity. Moreover, the time complexity can grow exponentially with the number of sources. This is known as the ``curse of dimensionality'' in MDPs. In Section~\ref{sec:lyapunov}, we will develop a low-complexity online policy.

\section{DPP: A Low-Complexity Benchmark}\label{sec:lyapunov}
This section presents an online \textit{drift-plus-penalty} (DPP) policy. Instead of a brute-force search of optimal policies, the DPP policy makes decisions by solving a simple ``MinWeight'' problem. It converts the average-cost constraint satisfaction problem into the queue stability problem and greedily trades between queue stability and CAE minimization.
\vspace{-0.1in}
\subsection{Problem Transformation}
We first define a \textit{virtual queue} associated with the constraint in \eqref{problem:CMDP} in such a way that a stable virtual queue implies a constraint-satisfying result. Let $Z_t$ denote the virtual queue with update equation
\begin{align}
	Z_{t+1} = \max[Z_t - F_{\max}, 0] + f_t, \label{eq:virtual queue}
\end{align}
where $Z_0 = 0$, $f_t = \mathds{1}(A_t\neq 0)$. Herein, $F_{\max}$ acts as a virtual service rate, and $f_t$ acts as a virtual arrival process. If virtual queue $Z_t$ is \textit{mean rate stable}, then the constraint in \eqref{problem:CMDP} is satisfied with probability $1$\cite{neely2010stochastic}.

We can use the quadratic \textit{Lyapunov function} $\Gamma(Z_t) = \frac{1}{2}Z_t^2$ as a scalar measure of queue congestion. The (one-slot) conditional \textit{Lyapunov drift} for slot $t$ is defined as
\begin{align}
	\Delta(Z_t) \triangleq \mathbb{E}_\textrm{a} \big\{\Gamma(Z_{t+1}) - \Gamma(Z_{t})\big|Z_t\big\}\label{eq:drift function},
\end{align}
where the expectation is with respect to the (possibly random) sampling actions. By using the inequality $(\max[W-b, 0] + a)^2 \leq W^2 + a^2 + b^2 + 2W(a-b)$, we have that $\Delta(Z_t)$ for a general policy satisfies
\begin{align}
	\Delta(Z_t) \leq B + \mathbb{E}_\textrm{a} \bigl\{Z_t (f_t - F_{\max}) \big| Z_t\bigr\}, \label{eq:bound on drift}
\end{align}
where $B = (1 + F^2_{\max})/2$.

To stabilize the virtual queue and minimize the average CAE, we can seek a policy to minimize the upper bound on the following expression:
\begin{align}
	\mathbb{E}_\textrm{a} \bigl\{Z_t (f_t - F_{\max})\big|Z_t \big\} + V \mathbb{E}_\textrm{sa} \bigl\{c_t\big|Z_t \big\},\label{eq:dpp_expression}
\end{align}
where $c_t$ is the CAE at time $t$, $V\geq0$ is the weight factor that represents how much emphasis we put on CAE minimization. Notice that the expectation of the second term is with respect to the system dynamics and the (possibly random) actions. 
\begin{algorithm}[htbp]
    \renewcommand{\thealgocf}{3}
    \DontPrintSemicolon
    \SetAlgoLined
    \algsetup{linenosize=\small}
	\KwIn{$\lambda$, $c$, $f$}
	\KwOut{$\pi_\lambda^*, \hat{F}^\lambda, \hat{C}^\lambda$, $\hat{\mathcal{L}}^\lambda$}
	\BlankLine
	$s_\textrm{ref}=0, \bm{q}^0(s,a) =0$ for all $s\in\mathcal{S}, a\in\mathcal{A}$\;
	\For{$k = 0, 1, \ldots, K-1$}  {    
		\For{$(s,a) \in \mathcal{S}\times\mathcal{A}$}{   
			\textsc{Apply} action $a$ to the system and observe $s^\prime$.\;
			\textsc{Update} the Q-factor using Eq.~\ref{eq:q-learning}\;
		}
	}
	\For{each $s\in \mathcal{S}$}{    
			$\pi^*_\lambda(s)\leftarrow\argmin_{a\in \mathcal{A}}\bm{q}^k(s,a)$\;
	}
	\textsc{Apply} policy $\pi^*_\lambda$ to the system and calculate $\hat{\mathcal{L}}^\lambda, \hat{F}^\lambda$, and $\hat{C}^\lambda$ based on empirical average.\;
	\Return{$\pi_\lambda^*, \hat{F}^\lambda, \hat{C}^\lambda$, $\hat{\mathcal{L}}^\lambda$}\;
	\caption{Average-cost Q-learning for $\mathscr{P}_\textrm{L-MDP}^\lambda$.}
 \label{alg:QL_LMDP}
\end{algorithm}
\vspace{-0.2in}
\subsection{Drift-Plus-Penalty Policy}
The DPP method operates as follows. Let $\bar{c}_t$ denote the (one-step) expected CAE given by Lemma~\ref{lemma:expected CAE}. Every slot $t$, given current queue backlog $Z_t$ and system state $S_t$, it chooses an action by solving the following problem:
\begin{align}
	\min_{\alpha_t \in \mathcal{A}} ~ Z_t (f_t - F_{\max}) + V \bar{c}_t. \label{problem:DPP}
\end{align}
The proposed DPP method is summarized in Algorithm \ref{alg:dpp}. In Theorem~\ref{theorem:stability}, we show that the proposed DPP policy satisfies the transmission frequency constraint.

\begin{lemma}\label{lemma:expected CAE} The one-step expected CAE is given by
	\begin{align}
		\bar{c}_t = \mathbb{E}_\textrm{s} \bigl\{c_t \big| Z_t\bigr\} = \sum_{m \in \mathcal{M}} \omega_m  \bar{\delta}^m_t,
	\end{align}
	where the expectation is with respect to the system dynamics, $\bar{\delta}^m_t$ is the one-slot expected CAE of subsystem $m$. For one-delay case, $\bar{\delta}^m_t$ is given as
	\begin{align}
		\bar{\delta}^m_t = 
		\begin{cases}
			\sum\limits_{k\neq i} \delta^m_{k,i} \bm{Q}^m_{i,k} p_s 
			+ 
			\sum\limits_{k\neq j} \delta^m_{k,j} \bm{Q}^m_{i,k} (1-p_s), & \textrm{if}~\alpha_t = m, \\
			\sum\limits_{k\neq j} \delta^m_{k,j} \bm{Q}^m_{i,k}, & \textrm{if}~\alpha_t \neq m.
		\end{cases}\notag
	\end{align}
	where $i,j,k\in\mathcal{X}^m$. Similarly, for zero-delay case we have
	\begin{align}
		\bar{\delta}^m_t = 
		\begin{cases}
			\delta^m_{i,j} (1-p_s), & \textrm{if}~\alpha_t = m, \\
			\delta^m_{i,j}, & \textrm{if}~\alpha_t \neq m.
		\end{cases}\notag
	\end{align}
\end{lemma}

\begin{theorem}\label{theorem:stability}
	For any $V\geq 0$, the DPP policy satisfies the transmission frequency constraint in \eqref{problem:CMDP}.
\end{theorem}
\begin{proof}
It is sufficient to show that the virtual queue under the DPP policy is mean rate stable. Our proof relies on the SA policy discussed in Section~\ref{sec:problem_formulation} and the Lyapunov optimization theorem\cite[Theorem~4.2]{neely2010stochastic}. Consider an SA policy $\pi_\textrm{sa}$ with
\begin{align}
    f_m = \frac{F_{\max}-\sigma}{M},\quad m = 1, \ldots, M,\notag
\end{align}
where $\sigma\geq 0$ is a small constant. The SA policy takes action irrespective of the system state and the queue backlog. Moreover, it satisfies the following slackness assumption 
\begin{align}
    \mathbb{E}\{f_t(\pi_\textrm{sa}) - F_{\max}|Z_t\} = -\sigma.\notag
\end{align}

Implementing both the SA policy and the DPP policy yields:
\begin{align}
    &\Delta(Z_t) \hspace{-0.2em}+\hspace{-0.2em} V\bar{c}_t \notag
    \leq B \hspace{-0.2em}+\hspace{-0.2em} V\bar{c}_t(\pi_\textrm{dpp}) \hspace{-0.2em}+\hspace{-0.2em} \mathbb{E}\{Z_t(f_t(\pi_\textrm{dpp}) - F_{\max})|Z_t\} \notag\\
    &\overset{(a)}{\leq} B + V\bar{c}_t(\pi_\textrm{sa}) + Z_t \mathbb{E}\{f_t(\pi_\textrm{sa}) - F_{\max}|Z_t\}\notag
    \overset{(b)}{\leq} B^\prime -\sigma Z_t,\notag
\end{align}
where (a) holds because the DPP chooses the best action, including the one taken by the SA policy; (b) holds because the SA policy is independent of queue backlog and it satisfies the slackness assumption. It follows directly from \cite[Theorem~4.2]{neely2010stochastic} that $Z_t$ is mean rate stable.
\end{proof}
\vspace{-0.1in}
\begin{remark}
    The DPP policy has a low complexity of $\mathcal{O}(|\mathcal{A}|)$ and can support large-scale systems. However, DPP is sub-optimal because it greedily selects actions and ignores long-term performance. Moreover, it requires prior knowledge of the channel and source statistics to compute the expected costs of taking an action in a given state. Of course, one can apply an estimation of the system statistics and then use this approach, but with a penalty on the performance that will depend on the estimate's accuracy.
\end{remark}

A comparison of the proposed policies is provided in Table~\ref{table:complexity}, where $I_\textrm{RVI}$ and $I_\textrm{QL} (I_\textrm{RVI} \ll I_\textrm{QL})$ are the number of iterations required for the RVI and the average-cost Q-learning, respectively. We note that the time complexity of a policy includes two phases: (i) an offline phase (before deployment) that tunes the policy parameters using historical data or prior system information, and (ii) an online phase (after deployment) that interacts with the system and selects the best action using pre-determined rules or policies. The Insec-RVI and the RL algorithms exhibit obvious complexity disadvantages in the offline phase. In contrast, the DPP policy requires no computational expenditure in the offline phase and has low online time complexity. Therefore, \emph{there is a tradeoff between optimality and complexity}. In practice, one can choose a desired policy according to the availability of system dynamics and resources.

\begin{table}[htbp]
\centering
\scriptsize
\caption{Comparison of The Proposed Policies}
\begin{tabular}{ c | c | c | c | c}
    \hline
     \multirow{2}{*}{\textbf{Algorithm}}
     & \multirow{2}{*}{\textbf{Scalability}}
     & \multirow{2}{*}{\textbf{Optimality}} 
     & \multicolumn{2}{c}{\textbf{Time complexity}}\\
    \cline{4-5}
     & & & \textbf{Offline} & \textbf{Online}\\
    \hline
    \text{SA} & \checkmark & Sub-optimal & --- & $\mathcal{O}(1)$\\
    \text{Bisec-RVI} & $\times$ & Optimal & $\mathcal{O}(\log_2(\frac{\lambda_{\max}}{\epsilon})I_\textrm{RVI})$  & $\mathcal{O}(1)$\\
    \text{Insec-RVI} & $\times$ & Optimal & $\mathcal{O}\left((W-1)I_\textrm{RVI}\right)$  & $\mathcal{O}(1)$\\
    \text{RL} & $\times$  & Optimal & $\mathcal{O}\left((W-1)I_\textrm{QL}\right)$ &  $\mathcal{O}(1)$ \\
    \text{DPP} & \checkmark  & Sub-optimal & --- &   $\mathcal{O}(|\mathcal{A}|)$ \\
    \hline
\end{tabular}
\label{table:complexity}
\end{table}

\begin{algorithm}[htbp]
    \renewcommand{\thealgocf}{4}
    \DontPrintSemicolon
    \SetAlgoLined
    \algsetup{linenosize=\small}
    \caption{DPP policy for $\mathscr{P}_\textrm{CMDP}$.}
    \label{alg:dpp}
    Set $V$, and initialize virtual queue $Z_0 = 0$\;
    \For{$t = 1, 2, \ldots$}{
		\textsc{Observe} $S_t$ and $Z_t$ at the beginning of slot $t$.\; \textsc{Calculate} $\bar{c}_t$ using Lemma \ref{lemma:expected CAE}.\;
		\textsc{Select} an action $a^*$ according to \eqref{problem:DPP}.\;
		\textsc{Apply} action $a^*$. Update $Z_t$ using \eqref{eq:virtual queue} and $S_t$ according to \eqref{eq:Pm}\;
	}
\end{algorithm}

\begin{figure*}[htbp]
    \begin{subfigure}{0.33\linewidth}
        \centering
        \includegraphics[width=\linewidth]{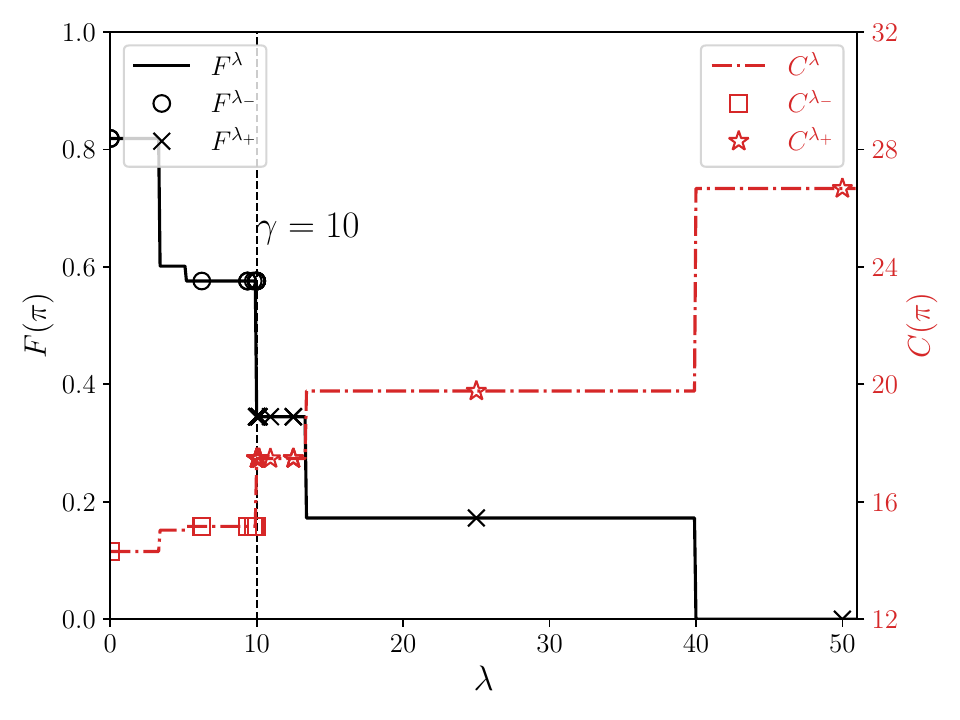}\vspace{-0.1in}
        \caption{Bisection search.}
        \label{fig:birviconvergence}
    \end{subfigure}
    \begin{subfigure}{0.33\linewidth}
        \centering
        \includegraphics[width=\linewidth]{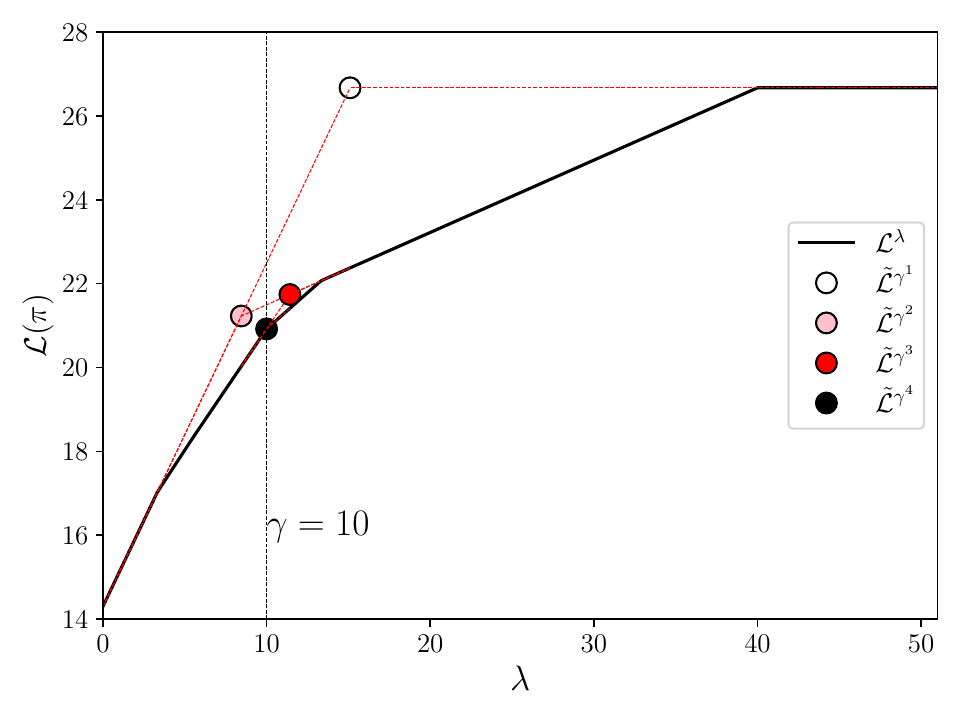}\vspace{-0.1in}
        \caption{Intersection search.}
        \label{fig:sarviconvergence}
    \end{subfigure}     
    \begin{subfigure}{0.33\linewidth}
	\centering
	\includegraphics[width=\linewidth]{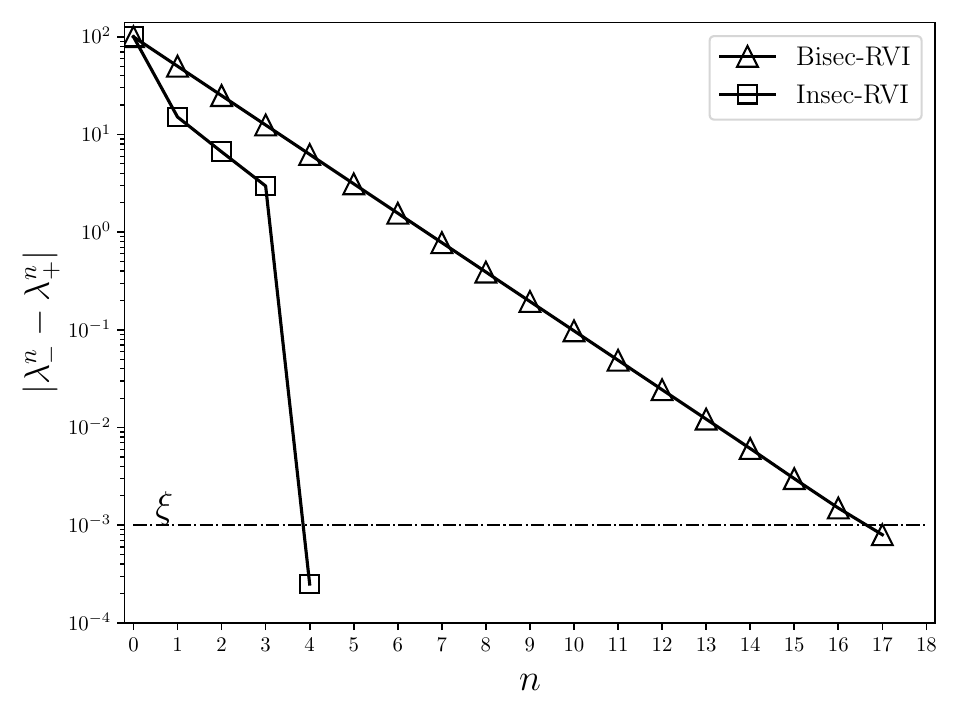}\vspace{-0.1in}
	\caption{Convergence performance.}
	\label{fig:comparebisa}
    \end{subfigure}
    \vspace{-0.2in}
    \caption{Comparison of different policy search methods when $F_{\max} = 0.4$ and $p_s = 0.4$. The optimal Lagrangian multiplier is found at $\lambda = 10$.}
    \label{fig:policysearch}
\end{figure*}
\begin{figure*}[htbp]
    \begin{subfigure}{0.33\linewidth}
        \centering
        \includegraphics[width=\linewidth]{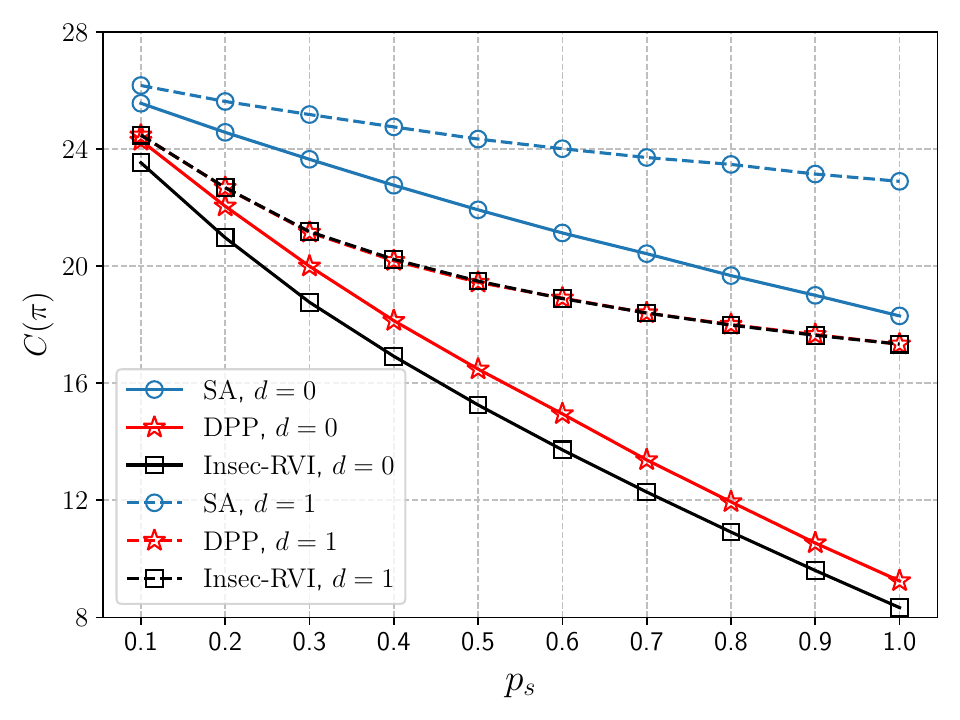}
        \caption{Avg. CAE vs. $p_s$ when $F_{\max} = 0.4$.}
        \label{fig:obj_ps}
    \end{subfigure}
    \begin{subfigure}{0.33\linewidth}
        \centering
        \includegraphics[width=0.99\linewidth]{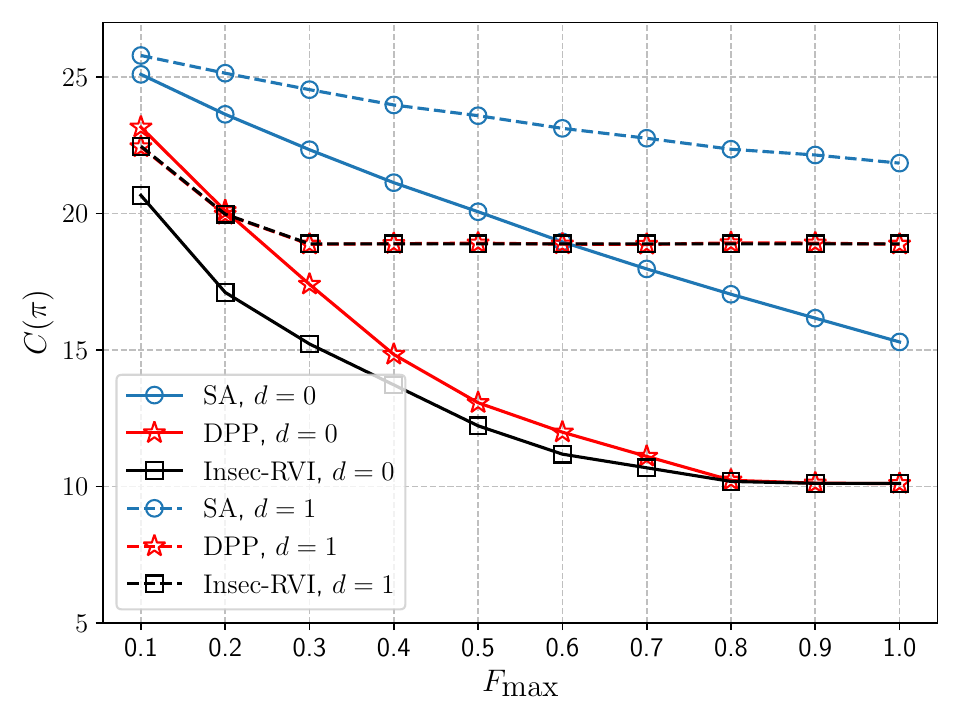}
        \caption{Avg. CAE vs. $F_{\max}$ when $p_s = 0.6$.}
        \label{fig:obj_cmax}
    \end{subfigure}  
    \begin{subfigure}{0.33\linewidth}
        \centering
        \includegraphics[width=\linewidth]{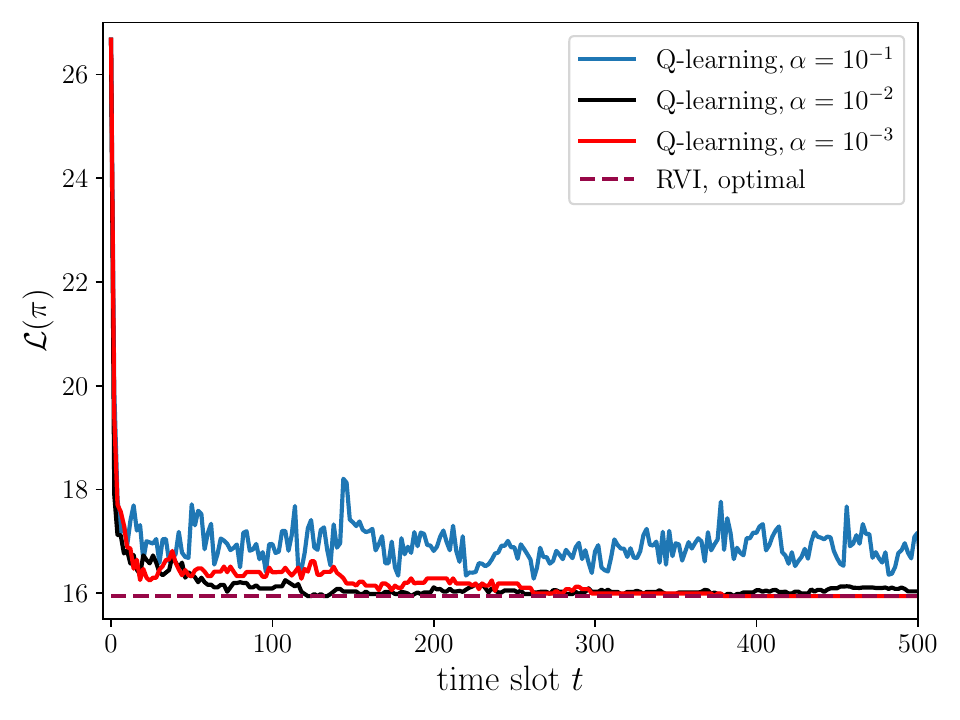}
        \caption{Convergence of the average-cost Q-learning.}
        \label{fig:qlearning-convergence}
    \end{subfigure}  
    \caption{Performance comparisons of different policies.}
    \label{fig:performancecomparison}
\end{figure*}
\begin{figure*}[htbp]
    \begin{subfigure}{0.33\linewidth}
         \centering
        \includegraphics[width=\linewidth]{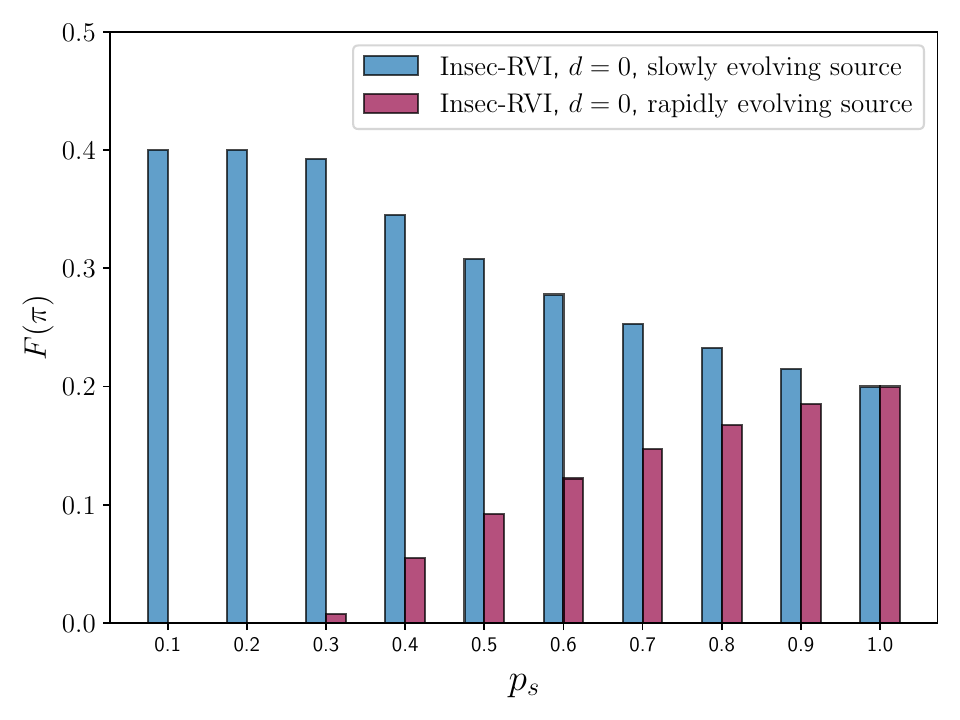}\vspace{-0.1in}
        \caption{Insec-RVI policy ($d=0$).}
        \label{fig:rvi_zero_delay_trans}
    \end{subfigure}
    \begin{subfigure}{0.33\linewidth}
         \centering
        \includegraphics[width=\linewidth]{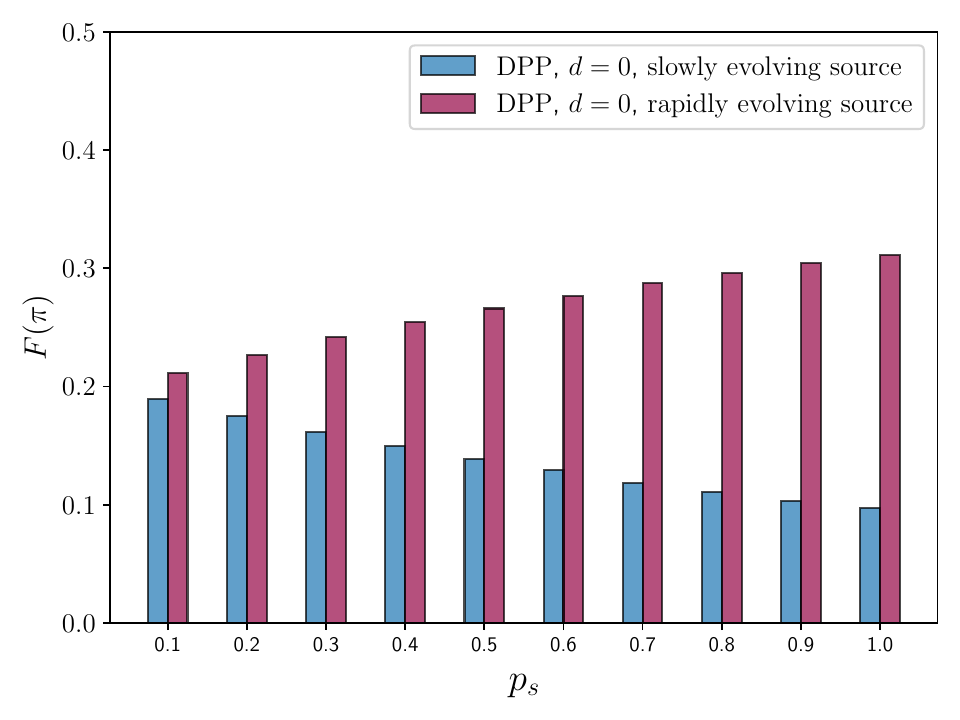}\vspace{-0.1in}
        \caption{DPP policy ($d=0$).}
        \label{fig:dpp_zero_delay_trans}
    \end{subfigure}  
    \begin{subfigure}{0.33\linewidth}
         \centering
        \includegraphics[width=\linewidth]{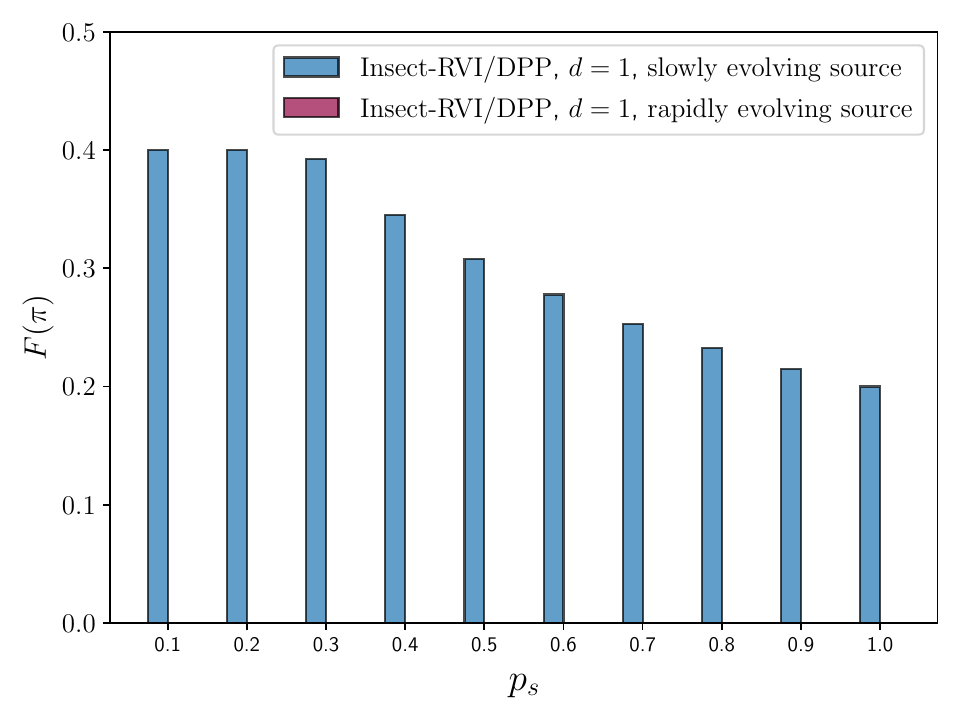}\vspace{-0.1in}
        \caption{Insec-RVI and DPP policies ($d=1$).}
        \label{fig:onedelaytrans}
    \end{subfigure}  
    \caption{Transmission frequency vs. $p_s$ when $F_{\max} = 0.4$.}
    \label{fig:behavioralanalysis}
\end{figure*}

\section{Numerical Results}\label{sec:results}
In this section, we validate the performance of the proposed methods. We consider an NCS with $M=2$ information sources, both assumed to be equally important, i.e., $\omega_1 = \omega_2 = 1$. Each source has three states and is characterized by a symmetric transition matrix, that is, $\mathbf{Q}_{i,j} = p$ for all $i \neq j$ and $1-2p$ otherwise. For comparison purposes, the first source is assumed to be slowly evolving with $p = 0.1$, while the second evolves rapidly with $p = 0.4$. The CAE matrices are given by
\begin{align*}
	\delta^1 = \delta^2 = \begin{bmatrix}
		0 & 10 & 30\\
		30 & 0 & 10\\
		10 & 30 & 0
	\end{bmatrix}.
\end{align*}

The convergence tolerance for the RVI algorithm and the Bisection search are $\epsilon = 10^{-2}$ and $\xi = 10^{-3}$, respectively. The initial search interval for the Bisec-RVI algorithm is set to $\lambda_{\max} = 100$. The weight factor for the DPP method is $V = 100$. We consider that the SA policy selects each source with equal probability, i.e., $f_1 = f_2 = F_{\max}/2$. 

\subsection{Convergence of the Lagrangian Multiplier Update Methods}\label{sec:sim-convergency}
Fig.~\ref{fig:policysearch} compares the performance of the Bisec-RVI and the Insec-RVI methods in the zero-delay case. It can be observed that $F^\lambda$ is piecewise constant and $\mathcal{L}^\lambda$ is piecewise linear and concave. Moreover, the optimal Lagrangian multiplier is found at $\lambda = 10$, which is a breakpoint of $F^\lambda$ and a corner of $\mathcal{L}^\lambda$. 

Interestingly, it can be seen from Fig.~\ref{fig:birviconvergence} that for any $F_{\max}\geq0.8185$ the solution to the unconstrained problem $\mathscr{P}_\textrm{L-MDP}^0$ is also a solution to the constrained problem $\mathscr{P}_\textrm{CMDP}$. It implies that \textit{continuous transmission is unnecessary and has little marginal performance gain.} Additionally, when $F_{\max}\in[0.8185, 1]\cup\{0.8185,0.6012,0.5758,0.3448,0.1724,0\}$, the optimal policy is Markovian deterministic. Otherwise, there is no such $\lambda$ that satisfies $F^\lambda = F_{\max}$. By Theorem~\ref{theorem:structure_of_optimal_policy}, the optimal policy is a combination of two policies in space $\Pi^\textrm{MD}$. 

Fig.~\ref{fig:comparebisa} plots the width of the search interval of the Bisec-RVI and the Insec-RVI methods. The Bisec-RVI method demonstrates a logarithmic decrease in the search interval, converging within 17 iterations in this particular case. However, it requires 13 iterations in the subinterval $[5.2, 13.4]$ to meet the convergence tolerance of $\xi = 10^{-3}$. Also, the number of iterations goes to infinity as $\epsilon\rightarrow0$. On the other hand, the Insec-RVI algorithm requires only $4$ iterations and it terminates immediately when reaching this subinterval.

\subsection{Performance Comparisons of Different Policies}\label{sec:sim-dpp}
Fig.~\ref{fig:obj_ps} shows the average CAE achieved for different success transmission probabilities $p_s$ and delays $d$. Clearly, the average CAE decreases as the channel condition improves. Also, the transmission delay has a significant impact on the system performance. An important observation is that the gap between the DPP policy and the Insec-RVI policy is approximately $1.07$ for the zero-delay case ($d = 0$) and $0$ for the one-delay case ($d = 1$). Therefore, the DPP policy is promising in practice as it achieves near-optimal performance with low computational complexity and good scalability. The SA policy, however, shows a big performance disparity compared to the Insec-RVI policy, especially in relatively reliable channel conditions. This is because the SA policy takes probabilistic actions regardless of the information content, whereas the semantic-aware policy exploits the significance of information to transmit the right packet at the right time. 

Fig.~\ref{fig:obj_cmax} depicts the average CAE achieved by different policies for different $F_{\max}$ and $d$ values. Clearly, the average CAE decreases as $F_{\max}$ increases. Notably, when $F_{\max}$ exceeds a specific value (e.g., $F_{\max}= 0.8185$ for the zero-delay case), the average CAE remains constant despite further increases in $F_{{\max}}$. This occurs because, beyond this threshold, the optimal policy is the one used for the unconstrained problem $\mathscr{P}^0_\textrm{L-MDP}$, as discussed in Fig.~\ref{fig:birviconvergence}. For the one-delay case, this threshold is lowered to $0.3$. This implies that channel delay diminishes the usefulness of certain status updates, resulting in reduced resource requirements. 

Fig.~\ref{fig:qlearning-convergence} shows the performance of the average-cost Q-learning for $\lambda = 2, F_{\max} = 0.4, p_s = 0.4$. It can be seen that a smaller learning rate results in a slower convergence rate but offers more stable and optimal performance. Conversely, a larger learning rate may cause more variations. The Q-learning algorithm with $\alpha = 10^{-3}$ achieves the same performance as the RVI in $400$ iterations. Therefore, it can estimate the optimal policy without prior knowledge of the system dynamics with an appropriately chosen learning rate. 

\subsection{Behavioral Analysis}\label{sec:sim-behavior}
Fig.~\ref{fig:rvi_zero_delay_trans} and Fig.~\ref{fig:dpp_zero_delay_trans} illustrate the transmission frequency of each source in the zero-delay scenario. Despite a small performance gap, they exhibit distinct behaviors. Evidently, the Insec-RVI policy schedules more frequently the slowly evolving source, whereas the DPP policy consistently prioritizes the rapidly evolving source. This could be attributed to the fact that the DPP policy selects the source with higher (one-step) expected costs, disregarding the long-term system performance. Since the slowly evolving source is more likely to remain synced, the DPP policy places more emphasis on the rapidly evolving source. Additionally, both the DPP and Insec-RVI policies increase the importance of the rapidly evolving source in relatively reliable channel conditions.

Fig.~\ref{fig:onedelaytrans} demonstrates the transmission frequency of each source in the one-delay scenario. A notable difference compared to the zero-delay case is that both the Insect-RVI and the DPP policies tend to neglect the rapidly evolving source. This distinction is owing to the ineffectiveness of delayed updates of the rapidly evolving source, as the source is more likely to transition to a different state at the time of actuation. Furthermore, we observe that the overall transmission frequency is lower than $F_{\max}$. It implies that \textit{we can attain optimal performance by strategically utilizing fewer transmissions and exploiting the timing of the important information.}

Now we look into how the Insec-RVI policy caters to state-dependent requirements. Let $T^1$ and $T^2$ denote the transmission frequencies of the slowly evolving and rapidly evolving sources obtained by the optimal policy. Here, $T^1_{i,j} (T^2_{i,j})$ represents the transmission frequency when the slowly (rapidly) evolving source is in state $i$ and the reconstructed state is $j$. When $d = 0$, $p_s = 0.4$, and $F_{{\max}} = 0.8$, the values of $T^1$ and $T^2$ are as obtained as
\begin{align*}
    T^1 \hspace{-3pt} = \hspace{-3pt} \begin{bmatrix}
        0      & 0.045  & 0.058 \\
        0.058  & 0      & 0.045 \\
        0.045  & 0.058  & 0
    \end{bmatrix}\hspace{-3pt},~
    T^2\hspace{-3pt} = \hspace{-3pt}\begin{bmatrix}
        0      & 0.066  & 0.098 \\
        0.098  & 0    & 0.066 \\
        0.066  & 0.098  & 0
    \end{bmatrix}\hspace{-3pt}.
\end{align*}
In this case, we observe a two-level significance. First, the rapidly evolving source is deemed more important than the slowly evolving one (as also demonstrated in Fig.~\ref{fig:rvi_zero_delay_trans} and Fig.~\ref{fig:onedelaytrans}). Second, certain states are scheduled more frequently, indicating that they are more critical than others. In contrast, the transmission frequencies under the distortion-optimal policy are given by
\begin{equation*}
    T^1 \hspace{-3pt} = \hspace{-3pt} \begin{bmatrix}
        0      & 0.058  & 0.058 \\
        0.058  & 0      & 0.058 \\
        0.058  & 0.058  & 0
    \end{bmatrix}\hspace{-3pt},~
    T^2 \hspace{-3pt} = \hspace{-3pt} \begin{bmatrix}
        0   & 0.076  & 0.076 \\
        0.076  & 0   & 0.076 \\
        0.076  & 0.076  & 0
    \end{bmatrix}\hspace{-3pt}.
\end{equation*}
It is clear that the distortion-optimal policy treats each source's states equally and is unable to meet the state-dependent and context-aware requirements in such systems. 

\section{Conclusions}\label{sec:conclusion}
We studied the semantic-aware remote estimation of multiple sources under a resource constraint. We characterized the existence and structure of a constrained optimal policy, showing that the optimal policy is never more complex than a randomized mixture of two stationary deterministic policies. By exploiting these results, we proposed an efficient intersection search method and developed the Insec-RVI and RL policies for scenarios with known and unknown system statistics, respectively. Additionally, we introduced the low-complexity DPP policy as a benchmark. Simulation results underscore that we can reduce the number of uninformative transmissions by exploiting the timing of important information. 

\bibliographystyle{IEEEtran}
\bibliography{abrv, ref}

\begin{thebibliography}{10}
\providecommand{\url}[1]{#1}
\csname url@samestyle\endcsname
\providecommand{\newblock}{\relax}
\providecommand{\bibinfo}[2]{#2}
\providecommand{\BIBentrySTDinterwordspacing}{\spaceskip=0pt\relax}
\providecommand{\BIBentryALTinterwordstretchfactor}{4}
\providecommand{\BIBentryALTinterwordspacing}{\spaceskip=\fontdimen2\font plus
\BIBentryALTinterwordstretchfactor\fontdimen3\font minus \fontdimen4\font\relax}
\providecommand{\BIBforeignlanguage}[2]{{%
\expandafter\ifx\csname l@#1\endcsname\relax
\typeout{** WARNING: IEEEtran.bst: No hyphenation pattern has been}%
\typeout{** loaded for the language `#1'. Using the pattern for}%
\typeout{** the default language instead.}%
\else
\language=\csname l@#1\endcsname
\fi
#2}}
\providecommand{\BIBdecl}{\relax}
\BIBdecl

\bibitem{jipingPIMRC2024}
J.~Luo and N.~Pappas, ``Goal-oriented estimation of multiple markov sources in resource-constrained systems,'' in \emph{Proc. IEEE Int. Symp. Pers. Indoor Mob. Radio Commun.}, 2024, pp. 1--6.

\bibitem{jipingWiOpt2024}
------, ``Semantic-aware remote estimation of multiple markov sources under constraints,'' in \emph{Proc. Int. Symp. Modeling Optim. Mobile, Ad hoc, Wireless Netw.}, 2024, pp. 15--21.

\bibitem{Aditya2017}
J.~Chakravorty and A.~Mahajan, ``Fundamental limits of remote estimation of autoregressive markov processes under communication constraints,'' \emph{{IEEE} Trans. Autom. Control}, vol.~62, no.~3, pp. 1109--1124, 2017.

\bibitem{NikolaosGoalOriented}
N.~Pappas and M.~Kountouris, ``Goal-oriented communication for real-time tracking in autonomous systems,'' in \emph{Proc. IEEE Int. Conf. Auto. Syst.}, 2021.

\bibitem{TIISurvey}
L.~Zhang, H.~Gao, and O.~Kaynak, ``Network-induced constraints in networked control systems—a survey,'' \emph{{IEEE} Trans. Ind. Informat.}, vol.~9, no.~1, pp. 403--416, 2013.

\bibitem{gielis2022critical}
J.~Gielis, A.~Shankar, and A.~Prorok, ``A critical review of communications in multi-robot systems,'' \emph{Curr. Robot. Rep.}, vol.~3, no.~4, pp. 213--225, 2022.

\bibitem{Marios}
M.~Kountouris and N.~Pappas, ``Semantics-empowered communication for networked intelligent systems,'' \emph{{IEEE} Commun. Mag.}, vol.~59, no.~6, pp. 96--102, 2021.

\bibitem{Karl2013RE}
J.~Wu, Q.-S. Jia, K.~H. Johansson, and L.~Shi, ``Event-based sensor data scheduling: Trade-off between communication rate and estimation quality,'' \emph{{IEEE} Trans. Autom. Control}, vol.~58, no.~4, pp. 1041--1046, 2013.

\bibitem{leong2020deep}
A.~S. Leong, A.~Ramaswamy, D.~E. Quevedo, H.~Karl, and L.~Shi, ``Deep reinforcement learning for wireless sensor scheduling in cyber--physical systems,'' \emph{Automatica}, vol. 113, p. 108759, 2020.

\bibitem{pezzutto2022transmission}
M.~Pezzutto, L.~Schenato, and S.~Dey, ``Transmission power allocation for remote estimation with multi-packet reception capabilities,'' \emph{Automatica}, vol. 140, p. 110257, 2022.

\bibitem{Cocco2023}
G.~Cocco, A.~Munari, and G.~Liva, ``Remote monitoring of two-state markov sources via random access channels: an information freshness vs. state estimation entropy perspective,'' \emph{{IEEE} J. Sel. Areas Inf. Theory.}, vol.~4, pp. 651--666, 2023.

\bibitem{Petar}
P.~Popovski \emph{et~al.}, ``A perspective on time toward wireless 6g,'' \emph{Proc. {IEEE}}, vol. 110, no.~8, 2022.

\bibitem{pappas2023age}
N.~Pappas, M.~A. Abd-Elmagid, B.~Zhou, W.~Saad, and H.~S. Dhillon, \emph{Age of Information: Foundations and Applications}.\hskip 1em plus 0.5em minus 0.4em\relax Cambridge Univ. Press, 2023.

\bibitem{WiSwarm}
V.~Tripathi \emph{et~al.}, ``Wiswarm: Age-of-information-based wireless networking for collaborative teams of uavs,'' in \emph{Proc. IEEE Int. Conf. Comput. Commun.}, 2023, pp. 1--10.

\bibitem{kutsevol2023experimental}
P.~Kutsevol, O.~Ayan, N.~Pappas, and W.~Kellerer, ``Experimental study of transport layer protocols for wireless networked control systems,'' in \emph{Proc. IEEE Int. Conf. Sens. Commun. Netw.}, 2023, pp. 438--446.

\bibitem{jayanth2023distortion}
S.~Jayanth, N.~Pappas, and R.~V. Bhat, ``Distortion minimization with age of information and cost constraints,'' in \emph{Proc. Int. Symp. Modeling Optim. Mobile, Ad hoc, Wireless Netw.}, 2023, pp. 1--8.

\bibitem{VoI}
O.~Ayan, M.~Vilgelm, M.~Kl\"{u}gel, S.~Hirche, and W.~Kellerer, ``Age-of-information vs. value-of-information scheduling for cellular networked control systems,'' in \emph{Proc. ACM/IEEE Int. Conf. Cyber-Phys. Syst.}, 2019, p. 109–117.

\bibitem{AoI_onur}
O.~Ayan, M.~Vilgelm, and W.~Kellerer, ``Optimal scheduling for discounted age penalty minimization in multi-loop networked control,'' in \emph{Proc. IEEE Consum. Commun. Netw. Conf.}, 2020, pp. 1--7.

\bibitem{AoI_estimation}
Y.~Sun, Y.~Polyanskiy, and E.~Uysal, ``Sampling of the wiener process for remote estimation over a channel with random delay,'' \emph{{IEEE} Trans. Inf. Theory}, vol.~66, no.~2, pp. 1118--1135, 2019.

\bibitem{AoII_TWC}
A.~Maatouk, M.~Assaad, and A.~Ephremides, ``The age of incorrect information: An enabler of semantics-empowered communication,'' \emph{{IEEE} Trans. Wireless Commun.}, vol.~22, no.~4, pp. 2621--2635, 2023.

\bibitem{AoII_GC}
Y.~Chen and A.~Ephremides, ``Minimizing age of incorrect information for unreliable channel with power constraint,'' in \emph{Proc. IEEE Global Commun. Conf.}, 2021, pp. 1--6.

\bibitem{niko2019statechange}
G.~Stamatakis, N.~Pappas, and A.~Traganitis, ``Control of status updates for energy harvesting devices that monitor processes with alarms,'' in \emph{Proc. IEEE Global Commun. Conf. Workshops}, 2019, pp. 1--6.

\bibitem{UoI_estimation}
X.~Zheng, S.~Zhou, and Z.~Niu, ``Urgency of information for context-aware timely status updates in remote control systems,'' \emph{{IEEE} Trans. Wireless Commun.}, vol.~19, no.~11, pp. 7237--7250, 2020.

\bibitem{AoA}
A.~Nikkhah, A.~Ephremides, and N.~Pappas, ``Age of actuation in a wireless power transfer system,'' in \emph{Proc. IEEE Int. Conf. Comput. Commun. Workshops}, 2023, pp. 1--6.

\bibitem{Salimnejad2023TCOM}
M.~Salimnejad, M.~Kountouris, and N.~Pappas, ``Real-time reconstruction of markov sources and remote actuation over wireless channels,'' \emph{{IEEE} Trans. Commun.}, pp. 1--1, 2024.

\bibitem{Salimnejad2023JCN}
------, ``State-aware real-time tracking and remote reconstruction of a markov source,'' \emph{Journal Commun. and Netw.}, vol.~25, no.~5, pp. 657--669, 2023.

\bibitem{EmmanouilMinizationCAE}
E.~Fountoulakis, N.~Pappas, and M.~Kountouris, ``Goal-oriented policies for cost of actuation error minimization in wireless autonomous systems,'' \emph{{IEEE} Commun. Lett.}, vol.~27, no.~9, pp. 2323--2327, 2023.

\bibitem{zakerisemantic}
A.~Zakeri, M.~Moltafet, and M.~Codreanu, ``Semantic-aware sampling and transmission in real-time tracking systems: A pomdp approach,'' \emph{{IEEE} Trans. Commun.}, pp. 1--1, 2024.

\bibitem{altman1999constrained}
E.~Altman, \emph{Constrained Markov decision processes}.\hskip 1em plus 0.5em minus 0.4em\relax CRC Press, 1999.

\bibitem{althoff2009}
M.~Althoff, O.~Stursberg, and M.~Buss, ``Model-based probabilistic collision detection in autonomous driving,'' \emph{{IEEE} Trans. Intell. Transp. Syst.}, vol.~10, no.~2, pp. 299--310, 2009.

\bibitem{ure2019}
N.~K. Ure, M.~U. Yavas, A.~Alizadeh, and C.~Kurtulus, ``Enhancing situational awareness and performance of adaptive cruise control through model predictive control and deep reinforcement learning,'' in \emph{Proc. IEEE Intell. Veh. Symp.}, 2019, pp. 626--631.

\bibitem{shi2017cyber}
D.~Shi, R.~J. Elliott, and T.~Chen, ``On finite-state stochastic modeling and secure estimation of cyber-physical systems,'' \emph{{IEEE} Trans. Autom. Control}, vol.~62, no.~1, pp. 65--80, 2017.

\bibitem{lipski2024age}
M.~Lipski, C.~Kam, S.~Kompella, and T.~Ephremides, ``Age of channel state information for collaborative beamforming,'' in \emph{Proc. ACM MobiHoc}, 2024, pp. 392--397.

\bibitem{puterman1994markov}
M.~L. Puterman, \emph{Markov decision processes: discrete stochastic dynamic programming}.\hskip 1em plus 0.5em minus 0.4em\relax John Wiley \& Sons, 1994.

\bibitem{sennott1998stochastic}
L.~I. Sennott, \emph{Stochastic dynamic programming and the control of queueing systems}.\hskip 1em plus 0.5em minus 0.4em\relax John Wiley \& Sons, 1998.

\bibitem{gallager1997discrete}
R.~G. Gallager, \emph{Discrete stochastic processes}.\hskip 1em plus 0.5em minus 0.4em\relax Taylor \& Francis, 1997.

\bibitem{bertsekas2011dynamic}
D.~P. Bertsekas, \emph{Dynamic Programming and Optimal Control, Volume II, 3rd edition}.\hskip 1em plus 0.5em minus 0.4em\relax Athena Scientific, 2007.

\bibitem{beutler1985optimal}
F.~J. Beutler and K.~W. Ross, ``Optimal policies for controlled markov chains with a constraint,'' \emph{J. Math. Anal. Appl.}, vol. 112, no.~1, pp. 236--252, 1985.

\bibitem{sennott_1993}
L.~I. Sennott, ``Constrained average cost markov decision chains,'' \emph{Probab. Eng. Inf. Sci.}, vol.~7, no.~1, p. 69–83, 1993.

\bibitem{rcpo}
C.~Tessler, D.~J. Mankowitz, and S.~Mannor, ``Reward constrained policy optimization,'' in \emph{Proc. Int. Conf. Learn. Represent. (ICLR)}, 2019.

\bibitem{djonin2007Q}
D.~V. Djonin and V.~Krishnamurthy, ``${Q}$-learning algorithms for constrained markov decision processes with randomized monotone policies: Application to mimo transmission control,'' \emph{{IEEE} Trans. Signal Process.}, vol.~55, no.~5, pp. 2170--2181, 2007.

\bibitem{Ma1986}
D.-j. Ma, A.~M. Makowski, and A.~Shwartz, ``Estimation and optimal control for constrained markov chains,'' in \emph{Proc. {IEEE} Conf. Decision and Control}, 1986, pp. 994--999.

\bibitem{bisection_robots}
M.~Ono, M.~Pavone, Y.~Kuwata, and J.~Balaram, ``Chance-constrained dynamic programming with application to risk-aware robotic space exploration,'' \emph{Auton. Robots}, vol.~39, pp. 555--571, 2015.

\bibitem{PDO2018}
Y.~Chow, M.~Ghavamzadeh, L.~Janson, and M.~Pavone, ``Risk-constrained reinforcement learning with percentile risk criteria,'' \emph{J. Mach. Learn. Res.}, vol.~18, no. 167, pp. 1--51, 2018.

\bibitem{lima2022model}
V.~Lima, M.~Eisen, K.~Gatsis, and A.~Ribeiro, ``Model-free design of control systems over wireless fading channels,'' \emph{Signal Process.}, vol. 197, p. 108540, 2022.

\bibitem{rviQlearning}
J.~Abounadi, D.~Bertsekas, and V.~S. Borkar, ``Learning algorithms for markov decision processes with average cost,'' \emph{SIAM J. Control Optim.}, vol.~40, no.~3, pp. 681--698, 2001.

\bibitem{neely2010stochastic}
M.~J. Neely, \emph{Stochastic network optimization with application to communication and queueing systems}.\hskip 1em plus 0.5em minus 0.4em\relax Morgan \& Claypool Publishers, 2010.

\end{thebibliography}

\end{document}